\documentclass{article}
\usepackage{amsmath,amssymb,amsthm,mathtools,mathrsfs}
\usepackage{graphicx}
\usepackage{newpxtext}

\newtheorem{theorem}{Theorem}[section]
\newtheorem{lemma}{Lemma}[section]
\theoremstyle{remark}
\newtheorem{remark}{Remark}[section]
\title{Basket implied volatility skew and stickiness\thanks{This work arose from a joint research project involving Mizuho Bank, Ltd., Mizuho-DL Financial Technology Co., Ltd., Mizuho Securities Co., Ltd.,  and The University of Osaka. The authors are grateful to the participants for helpful comments and discussions.
The authors used OpenAI's ChatGPT to assist with language editing and the checking of calculations and mathematical arguments. All AI-generated suggestions were critically reviewed and, where appropriate, revised by the authors, who take full responsibility for the content of this paper.
}}
\author{Masaaki Fukasawa$^\dagger$, Jun Maeda$^\ddagger$, and Tatsuya Ogiwara$^\ddagger$\\
{$^\dagger$ \small The University of Osaka}\\
{$^\ddagger$ \small Mizuho Securities Co., Ltd.}
}
\date{}

\begin{document}

\maketitle

\begin{abstract}
We study the short-maturity implied volatility and the
skew stickiness ratio for baskets of assets with continuous,
possibly rough, stochastic volatility. The fluctuation of the
instantaneous basket variance has two sources: fluctuations of
the constituent variances and fluctuations of the 
basket weights. We derive a near-the-money implied volatility
expansion that separates these contributions. We then
specialize the result to volatility models given by general
functions of Gaussian Volterra factors and obtain an explicit
basket-skew coefficient in terms of the short-time kernel
asymptotics, the factor sensitivities, and the return-factor
correlations. A density expansion justifies differentiation of the near-the-money
expansion at the money. Finally, using a Malliavin
representation of the dynamics of total implied variance, we
prove that the short-maturity skew stickiness ratio converges
to the universal limit $H+3/2$ for Gaussian factor basket
models with $H\in(0,1/2]$.
\end{abstract}

\section{Introduction}
\label{sec:introduction}

Equity indices and other baskets are weighted sums of assets
whose individual prices and volatilities evolve jointly. Even
when the constituent assets follow relatively simple
continuous stochastic volatility models, the volatility of the
basket has a more complicated structure. It depends not only on
the constituent volatilities and their correlations, but also
on the stochastic portfolio weights generated by relative
movements of the constituent prices. Understanding how these
two mechanisms contribute to the short-maturity implied
volatility skew is important both theoretically and for the
modeling of index options.

For a single asset, the short-maturity behavior of the
at-the-money implied volatility skew is closely related to the regularity of the spot variance. In classical
diffusion models, the skew remains finite as maturity
tends to zero. By contrast, rough volatility models, in which
the covariance between returns and spot variance has a
fractional term structure, produce a power law
which is consistent with the steep short-dated skews observed
in equity option markets. A general probabilistic explanation
of this relation was given by Fukasawa~\cite{F21}. 

For baskets, an additional effect appears. Let
\[
    S_t=\sum_{i=1}^nS_t^i,
    \qquad
    \Pi_t^i=\frac{S_t^i}{S_t},
\]
where $S^i_t$, $i=1,\dots,n$, are the constituent asset prices
at time $t$.
The instantaneous basket variance is
\[
    V_t
    =
    \sum_{i,j=1}^n
    \Pi_t^i\Pi_t^j
    \sqrt{V_t^iV_t^j}\rho^{ij},
\]
where $V^i$ is the spot variance process of $S^i$ and
$\rho^{ij}$ is the correlation between the Brownian motions
driving $S^i$ and $S^j$.
Its short-time fluctuation has two distinct sources. The first
comes from fluctuations of the constituent spot variances
$V_t^i$. If these fluctuations occur at a rate $r(t)$ that is
regularly varying with index $H$, their contribution is of
order $r(t)$. The second
comes from fluctuations of the endogenous basket weights
$\Pi_t^i$. Since the weights are driven directly by asset
returns, their contribution is always of order $\sqrt t$.
The two effects have different asymptotic origins and, when
$H<1/2$, different orders of magnitude.

The first purpose of this paper is to derive a general
short-maturity basket implied volatility expansion formula that separates these two
contributions under minimal conditions. Under weak convergence and uniform-integrability
conditions on the normalized constituent variance
fluctuations, we prove that the basket implied volatility $\sigma(k,\theta)$ admits
\begin{align*}
    \sigma(\sqrt\theta z,\theta)
    &=
    \sqrt{\bar v(\theta)}
    +
    z\left\{
        \frac{\psi_A}{2H+3}r(\theta)
        +
        \frac{\psi_B}{4}\sqrt\theta
    \right\}
    +
    o\bigl(r(\theta)+\sqrt\theta\bigr),
\end{align*}
where $\theta$ is time-to-maturity. In the pure-power case,
$r(\theta)=\theta^H$.
Here $\bar v(\theta)$ is the averaged deterministic
short-maturity basket variance. The coefficient $\psi_A$
describes the leverage between the basket return and the
constituent volatility fluctuations, whereas $\psi_B$ is the
contribution generated by fluctuations of the basket weights.

When $H<1/2$ and $\psi_A\neq0$, the volatility-fluctuation term dominates and the basket skew is regularly
varying with index $H-1/2$. In the pure-power case
$r(\theta)=\theta^H$, its leading order is
$\theta^{H-1/2}$.
The portfolio-weight term is then of smaller
order. At the boundary $H=1/2$, however, both terms must be retained in general. This boundary case is
particularly important for local and classical stochastic
volatility models. Omitting the portfolio-weight contribution
would in general give an incorrect finite short-maturity skew.

Short-maturity basket skews have previously been studied in
several specific settings. For local volatility models, related
formulas follow from the large-deviation approach of Avellaneda
et al.~\cite{A1,A2} and Pirjol~\cite{Pirjol}; see also Bayer and
Laurence~\cite{BL}. A formal derivation based on Markovian
projection was given by Piterbarg~\cite{Piterbarg}. Particular
multivariate stochastic volatility models include the
bivariate SABR model considered by Forde and
Zhang~\cite{FZ} and the normal SABR setting studied by Hagan et
al.~\cite{Hagan}. For rough Bergomi-type basket models, Friz and
Wagenhofer~\cite{FW} obtained basket-skew formulas using
small-noise large deviations. Their asymptotic regime is
closely related to, but different from, the short-time
near-the-money regime considered here. Our approach instead
starts from a general weak limit of normalized return and
variance fluctuations and therefore applies beyond a
particular Markovian, semimartingale, or exponential-Gaussian
specification.

We next specialize the general formula to Gaussian factor
volatility models. The constituent variances are allowed to be
general functions of time and finitely many Gaussian Volterra
factors.
The class of models incorporates the multi-factor Bergomi model~\cite{Bergomi}, the rough Bergomi model~\cite{RoughBook},
and the quintic model~\cite{AbiJaber}.
If the factor kernels behave near zero as
\[
    K(t)\sim \kappa\frac{r(t)}{\sqrt t},
\]
then a Gaussian delta method gives the limiting constituent
variance fluctuations explicitly. The resulting basket-skew
coefficient depends only on the logarithmic sensitivities of
the volatility functions at time zero, the short-time kernel
coefficients $(H,\kappa)$, the price-factor correlations, and
the initial basket weights and variances. 

The dynamics of the implied volatility surface are important in
addition to its static short-maturity shape. A commonly used
measure of smile dynamics is the skew stickiness ratio (SSR),
introduced by Bergomi~\cite{Bergomi}. It compares the
instantaneous response of at-the-money implied volatility to a
movement of the underlying with the contemporaneous
at-the-money implied volatility skew. 

There are different ways volatility can move relative to spot. For example, {\it vol-by-moneyness} assumes the volatility smile remains fixed as a function of moneyness, while {\it vol-by-strike} assumes the smile remains fixed as a function of absolute strike. These dynamics differ across markets and can also shift across regimes for the same underlying.
Because these dynamics affect how option values and Greeks evolve as spot moves, traders need to understand, and ideally anticipate, the prevailing volatility regime in order to hedge appropriately. The SSR provides one indicator for identifying which regime the underlying is in.

There are two reasons to focus on basket SSR.
First, it helps us understand the volatility dynamics of a basket. Equity indices such as the S\&P 500 and Nikkei 225 can be viewed as baskets of stocks. Since they have liquid listed options markets, we can observe index smiles directly and infer their volatility dynamics from market data. We also want to confirm that the dynamics observed at the index level are consistent with what is implied by the listed options on the index constituents.

Second, basket SSR is useful for hedging derivatives on a bespoke basket, where the relevant volatility regime (vol-by-moneyness, vol-by-strike, or another dynamic) is not directly observable. Computing SSR for the basket provides a practical indicator of how the smile is likely to shift, and therefore how to hedge derivative structures written on that basket.

The second purpose of this paper is therefore to determine the
short-maturity SSR for basket options. For a single asset,
Fukasawa~\cite{F26} showed, under Bergomi-type volatility
models, that the SSR converges to the universal limit
\[
    H+\frac32.
\]
For a basket, it is not immediate that the same limit should
hold. As shown by our basket-skew expansion, both fluctuations
of the constituent volatilities and fluctuations of the
endogenous portfolio weights affect the short-maturity skew.
These two sources also contribute to the response of ATM
implied volatility to movements of the basket, and their
relative importance changes at the boundary $H=1/2$.

We extend the single-asset analysis to Gaussian factor basket
models and prove, under a natural nondegeneracy condition, that
the basket SSR has the same
universal limit
\[
    H+\frac32
\]
for every $H\in(0,1/2]$. In particular, the limit remains valid
when $H=1/2$, at which the constituent-volatility and
portfolio-weight effects must in general both be retained.

The remainder of the paper is organized as follows.
Section~\ref{sec:basket-skew} establishes the general
short-maturity basket implied-volatility expansion.
Section~\ref{sec:Gaussian-factor-models} specializes the result
to Gaussian factor volatility models. Section~\ref{sec:stickiness} studies the basket
SSR and proves its universal short-maturity limit.

\section{Basket implied volatility skew}
\label{sec:basket-skew}
We assume frictionless markets with zero interest and dividend rates for brevity.

Consider $n$ assets
\begin{equation*}
    \frac{\mathrm{d}S^i_t}{S^i_t}
    =\sqrt{V^i_t}\,\mathrm dB^i_t,
    \qquad S_0^i>0,\quad i=1,\dots,n,
\end{equation*}
defined on a filtered probability space
$(\Omega,\mathscr{F},\mathsf{Q},\{\mathscr{F}_t\})$ with
\[
    \mathscr{F}_0
    =\{A\in\mathscr{F} ; \mathsf{Q}(A) \in \{0,1\}\},
\]
where $(B^1,\dots,B^n)$ is a constantly correlated
$\{\mathscr{F}_t\}$-Brownian motion with
\begin{equation*}
    \langle B^i, B^j \rangle_t = \rho^{ij} t
\end{equation*}
and the spot variance processes $V^i=\{V^i_t\}$
are $\{\mathscr{F}_t\}$-adapted and continuous nonnegative with $V^i_0 > 0$.

We are interested in the implied volatility $\sigma(k,\theta)$  of a basket option  defined through
\begin{equation*}
    \mathsf{E}[(S_0e^k-S_\theta)_+] = S_0p\left(k,\sigma(k,\theta)\sqrt{\theta}\right),
\end{equation*}
where
\begin{equation*}
    S_t = \sum_{i=1}^n S^i_t
\end{equation*}
and $p$ is the Black-Scholes put price:
\begin{equation*}
    p(k,w) = e^k\Phi(-d_-(k,w)) - \Phi(-d_+(k,w)),\ \ 
    d_\pm(k,w) = - \frac{k}{w} \pm \frac{w}{2}.
\end{equation*}
Here and hereafter, $\mathsf{E}$ refers to the expectation with respect to the measure $\mathsf{Q}$.
Our first main result is as follows.
\begin{theorem}\label{thm1}
Assume that the stochastic exponentials $S^i$ are true
martingales, and
that there exist $t_0>0$ and a function
$r:(0,t_0)\to(0,\infty)$, regularly varying at zero with index
$H\geq0$ and satisfying $r(0+)=0$, such that
\begin{equation}
Y^i = \{Y^i_t\}_{t \in (0,t_0)}, \quad
  Y^i_t :=  \frac{1}{r(t)}\left(\frac{V^i_t}{\mathsf{E}[V^i_t]} - 1 \right)
\end{equation}
are uniformly integrable  for all $i = 1,\dots, n$. 
Moreover, we assume that the joint law of
\[
    (X_\theta^1,\dots,X_\theta^n,
     Y_\theta^1,\dots,Y_\theta^n)
\]
converges to a centered $2n$-dimensional Gaussian distribution
as $\theta\to0$, where
\begin{equation*}
    X_\theta^i
    :=
\frac1{\sqrt{V^i_0 \theta}}\left(\frac{S^i_\theta}{S^i_0} - 1 \right).
\end{equation*}
Let $\psi^{ij}$ denote
the covariance between the limits of $X^i_\theta$ and
$Y^j_\theta$.
If $v(0)>0$,
then
\begin{align}
    \sigma(\sqrt{\theta}z,\theta)
    &=
    \sqrt{\bar v(\theta)}
    +
    \left\{
        \frac{\psi_A}{2H+3}r(\theta)
        +
        \frac{\psi_B}{4}\sqrt{\theta}
    \right\}z
    \nonumber\\
    &\quad+
    o\bigl(r(\theta)+\sqrt{\theta}\bigr)
    \label{basket-IV-expansion}
\end{align}
as $\theta\to0$, uniformly in $z$ on compact sets, where
\begin{equation*}
\begin{split}
    &\bar{v}(\theta) = \frac{1}{\theta}\int_0^\theta v(t) \mathrm{d}t, \\& v(t) = \sum_{i,j=1}^n v^{ij}(t), 
    \\ &v^{ij}(t) = \Pi^i_0\Pi^j_0 \mathsf{E}[\sqrt{V^i_t V^j_t}]\rho^{ij}, \\
    &\Pi^i_t = \frac{S^i_t}{S_t},
\end{split}
\end{equation*}
and
\begin{equation*}
\begin{split}
      \psi_A & =   \sum_{i,j,k=1}^n\frac{v^{ij}(0)}{2v(0)^{3/2}}(\psi^{ki} + \psi^{kj})\Pi^k_0\sqrt{V^k_0}, \\
    \psi_B & = \sum_{i,j,k=1}^n \frac{v^{ij}(0) }{v(0)^{3/2}}
  \left(\frac{v^{ik}(0)}{\Pi^i_0} + \frac{v^{jk}(0)}{\Pi^j_0}\right) -2\sqrt{v(0)}.
\end{split}
\end{equation*}
\end{theorem}

\begin{remark}[Comparison of the two rates]
\upshape
The ratio $r(\theta)/\sqrt\theta$ is regularly varying with
index $H-1/2$. Hence
\[
    \frac{r(\theta)}{\sqrt\theta}\longrightarrow\infty
    \quad\text{if }H<\frac12,
    \qquad
    \frac{r(\theta)}{\sqrt\theta}\longrightarrow0
    \quad\text{if }H>\frac12.
\]
When $H=1/2$, the ratio is slowly varying, and regular
variation alone does not imply that it converges: it may tend
to zero, a positive constant, or infinity, and it may also
oscillate. This is why the expansion retains both
$r(\theta)$ and $\sqrt\theta$ at the boundary rather than
replacing one by a fixed multiple of the other.
\end{remark}

\begin{remark}[Covariances with the basket]
\upshape
Define the initial covariance rate between the return of
constituent $i$ and the basket return by
\[
    d_i:=\sum_{j=1}^n\Pi_0^j\sqrt{V_0^iV_0^j}\rho^{ij}.
\]
Then $\sum_jv^{ij}(0)=\Pi_0^i d_i$ and
$v(0)=\sum_i\Pi_0^i d_i$. By the symmetry of $v^{ij}(0)$,
the coefficients in Theorem~\ref{thm1} can be written as
\begin{align*}
    \psi_A
    &=\frac{1}{v(0)^{3/2}}
      \sum_{i,k=1}^n\Pi_0^i\Pi_0^k d_i\sqrt{V_0^k}\,\psi^{ki},\\
    \psi_B
    &=\frac{2}{v(0)^{3/2}}
      \left\{\sum_{i=1}^n\Pi_0^i d_i^2-v(0)^2\right\}\\
    &=\frac{2}{v(0)^{3/2}}
      \sum_{i=1}^n\Pi_0^i\bigl(d_i-v(0)\bigr)^2\geq0.
\end{align*}
Thus the weight-fluctuation coefficient is proportional to
the variance of the $d_i$ under the initial basket weights.
\end{remark}

\begin{remark}[Local volatility]
\upshape
Suppose $V_t^i=\sigma_i(t,S_t^i)^2$, with positive bounded
$C^1$ local volatility functions having bounded first derivatives.
Then we may take $r(t)=\sqrt t$, so that $H=1/2$.
Write
\[
    s_i=\sigma_i(0,S_0^i),\qquad
    \ell_i=S_0^i\partial_s\log\sigma_i(0,S_0^i).
\]
The standard delta method gives
$(X_t^i,Y_t^i)_{i=1}^n\Rightarrow(G_i,2s_i\ell_iG_i)_{i=1}^n$,
where $G$ is centered Gaussian with covariance matrix
$(\rho^{ij})$. With $d_i=s_i\sum_j\Pi_0^j s_j\rho^{ij}$
as in the preceding remark, this gives
\begin{align*}
    \psi^{ki}&=2s_i\ell_i\rho^{ki},\\
    \psi_A&=\frac{2}{v(0)^{3/2}}
        \sum_{i=1}^n\Pi_0^i\ell_i d_i^2,\\
    \psi_B&=\frac{2}{v(0)^{3/2}}
        \sum_{i=1}^n\Pi_0^i\bigl(d_i-v(0)\bigr)^2,
\end{align*}
with $v(0)=\sum_i\Pi_0^i d_i>0$.
Thus the coefficient of $z\sqrt\theta$ in
\eqref{basket-IV-expansion} is $(\psi_A+\psi_B)/4$.
For a single asset, $\psi_B=0$ and this coefficient is
$S_0^1\partial_s\sigma_1(0,S_0^1)/2$.
\end{remark}

\begin{remark}[Preceding works]
For local volatility models, the above formula with $H=1/2$ is consistent with the one obtained by Pirjol~\cite{Pirjol}, based on the large-deviation theory developed by Avellaneda et al.~\cite{A1,A2}. See also Bayer and Laurence~\cite{BL}. The same formula also follows from a formal analysis by
Piterbarg~\cite{Piterbarg} based on the Markovian projection of stochastic volatility models when $H=1/2$.
A bivariate SABR model was treated by Forde and Zhang~\cite{FZ}.
A normal SABR model was treated by Hagan et al.~\cite{Hagan}.
For rough volatility models, that is, when $H<1/2$, Friz and Wagenhofer~\cite{FW} gave 
a small noise large deviation theory for rough Bergomi type models
with a basket-skew formula in a closely related but different asymptotic regime.
Our contribution is a rigorous derivation of the short-time asymptotic basket-skew formula for a general class of local and possibly rough stochastic volatility models.
Even in the non-rough case $H=1/2$, the result appears to be new.
\end{remark}

\begin{remark}[Constituent skew]
Let $\sigma^i(k,\theta)$ denote the implied volatility of $S^i$ defined through
\begin{equation*}
    \mathsf{E}[(S^i_0e^k-S^i_\theta)_+] = S^i_0p\left(k,\sigma^i(k,\theta)\sqrt{\theta}\right).
\end{equation*}
Under the assumptions of Theorem~\ref{thm1}, by Corollary 2.1 of Fukasawa~\cite{F21} (see also Theorem 8.11 of Bayer et al.~\cite{RoughBook}), we have
\begin{align*}
    \sigma^i(\sqrt{\theta}z,\theta)
    &=
    \sqrt{\bar v^i(\theta)}
    +
        \frac{\psi^{ii}}{2H+3}r(\theta) z
       +
    o(r(\theta))
\end{align*}
as $\theta\to0$ uniformly in $z$ on compact sets, where
\begin{equation*}
    \bar v^i(\theta) = \frac{1}{\theta} \int_0^\theta \mathsf{E}[V^i_t]\, \mathrm{d}t. 
\end{equation*}
\end{remark}

\begin{remark}[Dominance of the roughest constituent]
\upshape
The preceding result also indicates how the basket skew behaves
when the constituent volatility fluctuations have different
short-time regularities. Suppose, for simplicity, that the
normalized variance fluctuation of constituent $i$ is of order
$\theta^{H_i}$, where $H_i\in(0,1/2]$. Its contribution to the
ATM basket skew is then of order $\theta^{H_i-1/2}$.
Consequently, if
\[
    H_*:=\min_{1\leq i\leq n}H_i,
\]
then, provided the aggregate leverage coefficient associated
with the constituents satisfying $H_i=H_*$ does not vanish, the
basket skew satisfies
\[
    \frac{\sigma(k,\theta)-\sigma(0,\theta)}{k}
    \asymp
    \theta^{H_*-1/2}
\]
for $k = z \sqrt{\theta}$ with $z\neq 0$.
Thus the short-maturity explosion rate of the basket skew is in general the fastest among the explosion rates generated by
its constituents. 
\end{remark}

\begin{proof}
We first recall the result from
 Fukasawa~\cite{F21} (see also Bayer et al.~\cite{RoughBook})
that will be used below. Let
$a:(0,t_0)\to(0,\infty)$ be regularly varying at zero with index
$\alpha\geq 0$ and $a(0+)=0$. Suppose that, jointly,
\[
    \frac1{\sqrt\theta}
    \left(\frac{S_\theta}{S_0}-1\right)
    \Longrightarrow \xi,
    \qquad
    \frac1{a(\theta)}
    \left(
        \frac{V_\theta}{\mathsf E[V_\theta]}-1
    \right)
    \Longrightarrow \eta,
\]
assume that the second family is uniformly integrable and that
the limit $(\xi,\eta)$ is centered Gaussian. If
$V_0>0$, then
\begin{equation}\label{Fukasawa-expansion}
    \sigma(\sqrt\theta z,\theta)
    =
    \sqrt{\widehat v(\theta)}
    +
    \frac{\mathsf E[\xi\eta]}
         {\sqrt{V_0}(2\alpha+3)}
    z\,a(\theta)
    +
    o(a(\theta)),
\end{equation}
uniformly in $z$ on compact sets, where
\[
    \widehat v(\theta)
    :=
    \frac1\theta\int_0^\theta
    \mathsf E[V_t]\,\mathrm dt.
\]
The following steps 1-4 show the required uniform integrability.
The final step computes the limit distribution.

Put
\[
    m^i(t):=\mathsf E[V_t^i],
    \qquad
    A_t^{ij}:=\Pi_t^i\Pi_t^j,
    \qquad
    Q_t^{ij}:=\sqrt{V_t^iV_t^j},
\]
and
\[
    q^{ij}(t):=\mathsf E[Q_t^{ij}].
\]
Then
\[
    v^{ij}(t)
    =
    A_0^{ij}q^{ij}(t)\rho^{ij},
    \qquad
    v(t)=\sum_{i,j=1}^nv^{ij}(t).
\]

\medskip

\noindent
\textit{Step 1: Convergence of the deterministic variance
levels.}

By definition,
\[
    V_t^i
    =
    m^i(t)\bigl(1+r(t)Y_t^i\bigr),
    \qquad
    \mathsf E[Y_t^i]=0.
\]
Uniform integrability of $\{Y_t^i\}_{t\downarrow0}$ implies
that it is bounded in $L^1$. Therefore,
\[
    V_t^i-m^i(t)
    =
    r(t)m^i(t)Y_t^i.
\]
We first note that $m^i(t)\to V_0^i$. Indeed, continuity of
$V^i$ gives
\[
    V_t^i\longrightarrow V_0^i
    \qquad\text{in probability}.
\]
Moreover, since $Y_t^i$ is tight and $r(t)\to0$,
\[
    \frac{V_t^i}{m^i(t)}-1
    =
    r(t)Y_t^i
    \longrightarrow0
    \qquad\text{in probability}.
\]
It follows that
\[
    m^i(t)\longrightarrow V_0^i.
\]
Here we use that $m^i(t)$ and $V_0^i$ are deterministic and
$V_0^i>0$.

Consequently, after decreasing $t_0$ if necessary, there
exist constants $0<c<C<\infty$ such that
\[
    c\leq m^i(t)\leq C,
    \qquad
    0<t<t_0,
\]
for every $i$. In particular,
\[
    \sup_{0<t<t_0}\mathsf E[V_t^i]<\infty.
\]
Furthermore,
\[
    V_t^i-m^i(t)
    =
    r(t)m^i(t)Y_t^i
    \longrightarrow0
\]
in $L^1$. The family $\{V_t^i\}_{t\downarrow0}$ is therefore
uniformly integrable.

\medskip

\noindent
\textit{Step 2: Expansion of the geometric-mean volatility.}

We now prove that
\begin{equation}\label{Q-expansion-full}
    \frac{
        Q_t^{ij}-q^{ij}(t)
    }{r(t)}
    =
    \frac{\sqrt{m^i(t)m^j(t)}}{2}
    \bigl(Y_t^i+Y_t^j\bigr)
    +
    o_{L^1}(1).
\end{equation}

Since
\[
    V_t^i=m^i(t)\bigl(1+r(t)Y_t^i\bigr),
\]
we have
\begin{equation}\label{Q-factorization-full}
    Q_t^{ij}
    =
    \sqrt{m^i(t)m^j(t)}
    \sqrt{
        \bigl(1+r(t)Y_t^i\bigr)
        \bigl(1+r(t)Y_t^j\bigr)
    }.
\end{equation}
Notice that
\[
    1+r(t)Y_t^i\geq0,
    \qquad
    1+r(t)Y_t^j\geq0.
\]

Consider
\[
    f(x,y):=\sqrt{(1+x)(1+y)},
    \qquad x,y\geq-1.
\]
The function $f$ is differentiable at $(0,0)$, with
\[
    f(0,0)=1,
    \qquad
    \partial_xf(0,0)
    =
    \partial_yf(0,0)
    =
    \frac12.
\]
Thus, for every $M>0$,
\begin{equation}\label{compact-Taylor}
    \sup_{|x|+|y|\leq M}
    \left|
        \frac{
            f(r(t)x,r(t)y)-1
        }{r(t)}
        -
        \frac{x+y}{2}
    \right|
    \longrightarrow0.
\end{equation}

We also need a global bound. There exists a constant $C$
such that
\begin{equation}\label{global-square-root-bound}
    \left|
        \sqrt{(1+x)(1+y)}-1
    \right|
    \leq
    C\bigl(|x|+|y|\bigr),
    \qquad x,y\geq-1.
\end{equation}
To see this, if $|x|+|y|\leq1/2$, the assertion follows from
the mean value theorem, since the derivatives of $f$ are
bounded on this set. If $|x|+|y|>1/2$, then
\[
    \sqrt{(1+x)(1+y)}
    \leq
    \frac{(1+x)+(1+y)}2
    \leq
    1+\frac{|x|+|y|}{2}.
\]
The positive part of
$\sqrt{(1+x)(1+y)}-1$ is therefore bounded by
$(|x|+|y|)/2$, while its negative part is bounded by $1$,
which in the present region is bounded by
$2(|x|+|y|)$. This proves
\eqref{global-square-root-bound}.

Define
\[
    E_t^M
    :=
    \{|Y_t^i|+|Y_t^j|\leq M\}.
\]
On $E_t^M$, \eqref{compact-Taylor} gives
\begin{align*}
 &\sup_{\omega\in E_t^M}
 \left|
     \frac{
         \sqrt{
             (1+r(t)Y_t^i)(1+r(t)Y_t^j)
         }-1
     }{r(t)}
     -
     \frac{Y_t^i+Y_t^j}{2}
 \right|
 \longrightarrow0.
\end{align*}
On the complement, \eqref{global-square-root-bound} gives
\begin{align*}
 &\left|
     \frac{
         \sqrt{
             (1+r(t)Y_t^i)(1+r(t)Y_t^j)
         }-1
     }{r(t)}
     -
     \frac{Y_t^i+Y_t^j}{2}
 \right|
\\
 &\qquad\leq
 C\bigl(|Y_t^i|+|Y_t^j|\bigr).
\end{align*}
Consequently,
\begin{align*}
 &\mathsf E\left|
     \frac{
         \sqrt{
             (1+r(t)Y_t^i)(1+r(t)Y_t^j)
         }-1
     }{r(t)}
     -
     \frac{Y_t^i+Y_t^j}{2}
 \right|
\\
 &\quad\leq
 o_M(1)
 +
 C\mathsf E\left[
     \bigl(|Y_t^i|+|Y_t^j|\bigr)
     \mathbf 1_{\{
         |Y_t^i|+|Y_t^j|>M
     \}}
 \right].
\end{align*}
First letting $t\downarrow0$ and then $M\to\infty$, uniform
integrability of $Y^i$ and $Y^j$ gives
\begin{equation}\label{square-root-L1-expansion}
    \frac{
        \sqrt{
            (1+r(t)Y_t^i)(1+r(t)Y_t^j)
        }-1
    }{r(t)}
    =
    \frac12\bigl(Y_t^i+Y_t^j\bigr)
    +
    o_{L^1}(1).
\end{equation}

Multiplying \eqref{square-root-L1-expansion} by
$\sqrt{m^i(t)m^j(t)}$ and using
\eqref{Q-factorization-full}, we obtain
\begin{equation}\label{uncentered-Q-expansion}
    \frac{
        Q_t^{ij}-\sqrt{m^i(t)m^j(t)}
    }{r(t)}
    =
    \frac{\sqrt{m^i(t)m^j(t)}}2
    \bigl(Y_t^i+Y_t^j\bigr)
    +
    o_{L^1}(1).
\end{equation}
Since
\[
    \mathsf E[Y_t^i]
    =
    \mathsf E[Y_t^j]
    =
    0,
\]
taking expectations in \eqref{uncentered-Q-expansion} yields
\[
    q^{ij}(t)
    =
    \sqrt{m^i(t)m^j(t)}
    +
    o(r(t)).
\]
Subtracting expectations in
\eqref{uncentered-Q-expansion} proves
\eqref{Q-expansion-full}. It also follows that
\begin{equation}\label{normalized-Q-expansion-full}
    \frac{
        Q_t^{ij}-q^{ij}(t)
    }{
        r(t)q^{ij}(t)
    }
    =
    \frac12\bigl(Y_t^i+Y_t^j\bigr)
    +
    o_{L^1}(1).
\end{equation}
The family on the left-hand side is uniformly integrable as
$t\downarrow0$.

\medskip

\noindent
\textit{Step 3: Dynamics and first-order expansion of the
portfolio weights.}

By It\^o's formula,
\[
    \mathrm d\log S_t
    =
    \sum_{i=1}^n
    \Pi_t^i\frac{\mathrm dS_t^i}{S_t^i}
    -
    \frac12\mathrm d\langle\log S\rangle_t.
\]
Moreover,
\[
    \mathrm d\langle\log S\rangle_t
    =
    V_t\,\mathrm dt,
\]
where
\[
    V_t
    =
    \sum_{i,j=1}^n
    A_t^{ij}Q_t^{ij}\rho^{ij}.
\]

Since
\[
    \log A_t^{ij}
    =
    \log S_t^i+\log S_t^j-2\log S_t,
\]
It\^o's formula gives
\begin{align}
    \frac{\mathrm dA_t^{ij}}{A_t^{ij}}
    &=
    \sqrt{V_t^i}\,\mathrm dB_t^i
    +
    \sqrt{V_t^j}\,\mathrm dB_t^j
\nonumber\\
    &\quad
    -
    2\sum_{k=1}^n
    \Pi_t^k\sqrt{V_t^k}\,\mathrm dB_t^k
    +
    b_t^{ij}\,\mathrm dt,                  \label{A-dynamics-full}
\end{align}
where
\begin{equation}\label{b-bound-full}
    |b_t^{ij}|
    \leq
    C\sum_{k=1}^nV_t^k.
\end{equation}

Define
\[
    \bar X_\theta^i
    :=
    \frac1{\sqrt{V_0^i\theta}}
    \int_0^\theta
    \sqrt{V_t^i}\,\mathrm dB_t^i,
\]
and
\[
    \widetilde Z_\theta^{ij}
    :=
    \sqrt{V_0^i}\bar X_\theta^i
    +
    \sqrt{V_0^j}\bar X_\theta^j
    -
    2\sum_{k=1}^n
    \Pi_0^k\sqrt{V_0^k}\bar X_\theta^k.
\]
The constituent price dynamics give
\[
    X_\theta^i-\bar X_\theta^i
    =
    \frac1{\sqrt{V_0^i\theta}}
    \int_0^\theta
    \left(
        \frac{S_t^i}{S_0^i}-1
    \right)
    \sqrt{V_t^i}\,\mathrm dB_t^i.
\]
By continuity of $S^i$ and $V^i$,
    $X_\theta^i-\bar X_\theta^i
    \to 0$ 
 in probability.
Consequently, the $\bar X_\theta^i$ have the same joint limit
as the $X_\theta^i$. We claim that
\begin{equation}\label{weight-expansion-full}
    \frac1{\sqrt\theta}
    \left(
        \frac{A_\theta^{ij}}{A_0^{ij}}-1
    \right)
    =
    \widetilde Z_\theta^{ij}
    +
    o_{L^1}(1).
\end{equation}

Integrating \eqref{A-dynamics-full} and subtracting
$\widetilde Z_\theta^{ij}$ gives
\[
    \frac1{\sqrt\theta}
    \left(
        \frac{A_\theta^{ij}}{A_0^{ij}}-1
    \right)
    -
    \widetilde Z_\theta^{ij}
    =
    R_\theta^{ij,1}
    +
    R_\theta^{ij,2}
    +
    R_\theta^{ij,3},
\]
where
\begin{align*}
    R_\theta^{ij,1}
    &:=
    \frac1{\sqrt\theta}
    \int_0^\theta
    \left(
        \frac{A_t^{ij}}{A_0^{ij}}-1
    \right)
    \left(
        \sqrt{V_t^i}\,\mathrm dB_t^i
        +
        \sqrt{V_t^j}\,\mathrm dB_t^j
    \right.
\\
    &\hspace{4.8cm}\left.
        -
        2\sum_{k=1}^n
        \Pi_t^k\sqrt{V_t^k}\,\mathrm dB_t^k
    \right),
\end{align*}
\[
    R_\theta^{ij,2}
    :=
    -
    \frac2{\sqrt\theta}
    \sum_{k=1}^n
    \int_0^\theta
    \bigl(\Pi_t^k-\Pi_0^k\bigr)
    \sqrt{V_t^k}\,\mathrm dB_t^k,
\]
and
\[
    R_\theta^{ij,3}
    :=
    \frac1{\sqrt\theta}
    \int_0^\theta
    \frac{A_t^{ij}}{A_0^{ij}}
    b_t^{ij}\,\mathrm dt.
\]
Here we used
\[
    \frac{A_t^{ij}}{A_0^{ij}}\Pi_t^k-\Pi_0^k
    =
    \left(
        \frac{A_t^{ij}}{A_0^{ij}}-1
    \right)\Pi_t^k
    +
    \bigl(\Pi_t^k-\Pi_0^k\bigr).
\]

Because
\[
    0\leq A_t^{ij}\leq1,
    \qquad
    0\leq\Pi_t^k\leq1,
\]
continuity gives
\[
    A_t^{ij}\longrightarrow A_0^{ij},
    \qquad
    \Pi_t^k\longrightarrow\Pi_0^k
\]
in probability. Uniform integrability of $\{V_t^k\}_t$
therefore implies
\[
    \mathsf E\left[
        |A_t^{ij}-A_0^{ij}|^2V_t^k
    \right]
    \longrightarrow0
\]
and
\[
    \mathsf E\left[
        |\Pi_t^k-\Pi_0^k|^2V_t^\ell
    \right]
    \longrightarrow0.
\]
The Burkholder--Davis--Gundy inequality and Ces\`aro averaging
give
\[
    R_\theta^{ij,1}\longrightarrow0,
    \qquad
    R_\theta^{ij,2}\longrightarrow0
\]
in $L^1$. Also, by \eqref{b-bound-full},
\[
    \mathsf E|R_\theta^{ij,3}|
    \leq
    \frac{C}{\sqrt\theta}
    \int_0^\theta
    \sum_{k=1}^n\mathsf E[V_t^k]\,\mathrm dt
    =
    O(\sqrt\theta).
\]
This proves \eqref{weight-expansion-full}.

The same argument also yields
\begin{equation*}
    \mathsf E|A_t^{ij}-A_0^{ij}|
    \leq C\sqrt t.
\end{equation*}

\medskip

\noindent
\textit{Step 4: Two-scale expansion of the basket variance.}

Decompose
\begin{equation}\label{basket-V-decomposition-full}
 V_t=\sum_{i,j=1}^n A_0^{ij}Q_t^{ij}\rho^{ij}+D_t,
 \qquad
 D_t:=\sum_{i,j=1}^n(A_t^{ij}-A_0^{ij})Q_t^{ij}\rho^{ij}.
\end{equation}
For each $i,j$,
\begin{align}
 (A_t^{ij}-A_0^{ij})Q_t^{ij}
 &=q^{ij}(t)(A_t^{ij}-A_0^{ij})
 +(A_t^{ij}-A_0^{ij})(Q_t^{ij}-q^{ij}(t)).
 \nonumber
\end{align}
By \eqref{weight-expansion-full},
\begin{align}
 q^{ij}(\theta)(A_\theta^{ij}-A_0^{ij})
 &=\sqrt\theta\,A_0^{ij}q^{ij}(\theta)
   \widetilde Z_\theta^{ij}+o_{L^1}(\sqrt\theta).
   \nonumber
\end{align}
On the other hand, $A_\theta^{ij}-A_0^{ij}$ is bounded and
converges to zero in probability, whereas
\[
 \left\{
  \frac{Q_\theta^{ij}-q^{ij}(\theta)}{r(\theta)}
 \right\}_{\theta\downarrow0}
\]
is uniformly integrable by Step~2. Hence
\begin{equation*}
 (A_\theta^{ij}-A_0^{ij})(Q_\theta^{ij}-q^{ij}(\theta))
 =o_{L^1}(r(\theta)).
\end{equation*}
Consequently,
\begin{align}
 D_\theta
 &=\sqrt\theta\sum_{i,j=1}^n
 v^{ij}(\theta)\widetilde Z_\theta^{ij}
 +o_{L^1}(r(\theta))+o_{L^1}(\sqrt\theta).
 \label{D-two-scale}
\end{align}
Since $\mathsf E[\bar{X}_\theta^i]=0$, we have
$\mathsf E[\widetilde Z_\theta^{ij}]=0$. Taking expectations in
\eqref{basket-V-decomposition-full} and using
\eqref{D-two-scale}, we obtain
\begin{equation}\label{EV-two-scale}
 \mathsf E[V_\theta]
 =v(\theta)+o\bigl(r(\theta)+\sqrt\theta\bigr).
\end{equation}

Define
\[
 \eta_{A,\theta}:=
 \frac1{2v(\theta)}\sum_{i,j=1}^n
 v^{ij}(\theta)(Y_\theta^i+Y_\theta^j)
\]
and
\[
 \eta_{B,\theta}:=
 \frac1{v(\theta)}\sum_{i,j=1}^n
 v^{ij}(\theta)\widetilde Z_\theta^{ij}.
\]
Subtracting expectations in \eqref{basket-V-decomposition-full},
and using \eqref{Q-expansion-full}, \eqref{normalized-Q-expansion-full},  \eqref{D-two-scale} and
\eqref{EV-two-scale}, gives
\begin{align}
 \frac{V_\theta}{\mathsf E[V_\theta]}-1
 &=r(\theta)\eta_{A,\theta}
 +\sqrt\theta\eta_{B,\theta}
 +o_{L^1}(r(\theta))+o_{L^1}(\sqrt\theta).
 \label{normalized-V-two-scale}
\end{align}

Both $\{\eta_{A,\theta}\}$ and $\{\eta_{B,\theta}\}$ are
uniformly integrable. The first assertion follows from the
assumed uniform integrability of the $Y_\theta^i$. For the
second, observe that
\[
 \sup_{0<\theta<\theta_0}\mathsf E|\bar X_\theta^i|^2
 =\sup_{0<\theta<\theta_0}\frac1\theta\int_0^\theta
 \frac{\mathsf E[V_t^i]}{V_0^i}\,\mathrm dt<\infty.
\]
Thus the $\bar X_\theta^i$, and hence the
$\widetilde Z_\theta^{ij}$ and $\eta_{B,\theta}$, are uniformly
integrable.

\medskip

\noindent
\textit{Step 5: Application of the short-time implied-volatility
expansion.}

Let
\[
    (X^1,\dots,X^n,Y^1,\dots,Y^n)
\]
denote the centered Gaussian limit of
\[
    (X_\theta^1,\dots,X_\theta^n,
     Y_\theta^1,\dots,Y_\theta^n),
\]
and define
\[
    \xi_\theta
    :=
    \sum_{k=1}^n
    \Pi_0^k\sqrt{V_0^k}X_\theta^k,
    \qquad
    \xi
    :=
    \sum_{k=1}^n
    \Pi_0^k\sqrt{V_0^k}X^k.
\]
By the definition of $X_\theta^k$ and the identity
$\Pi_0^k=S_0^k/S_0$, we have the exact relation
\begin{equation}\label{return-xi-relation}
\begin{split}
    \xi_\theta
    =
    \frac1{\sqrt\theta}
    \sum_{k=1}^n
    \frac{S_0^k}{S_0}
    \left(
        \frac{S_\theta^k}{S_0^k}-1
    \right)
    =
    \frac1{\sqrt\theta}
    \left(
        \frac{S_\theta}{S_0}-1
    \right).
\end{split}
\end{equation}
Consequently, $\xi_\theta\Longrightarrow\xi$.
Moreover,
\[
    \mathsf E[\xi^2]
    =
    \sum_{k,\ell=1}^n
    \Pi_0^k\Pi_0^\ell
    \sqrt{V_0^kV_0^\ell}\rho^{k\ell}
    =
    v(0).
\]

By \eqref{normalized-V-two-scale}, the limiting volatility-fluctuation variable is
\[
 \eta_A:=\frac1{2v(0)}\sum_{i,j=1}^n
 v^{ij}(0)(Y^i+Y^j).
\]
Therefore,
\begin{align}
 \mathsf E[\xi\eta_A]
 &=\frac1{2v(0)}\sum_{i,j,k=1}^n
 v^{ij}(0)\Pi_0^k\sqrt{V_0^k}
 (\psi^{ki}+\psi^{kj})
 =\sqrt{v(0)}\psi_A.
 \label{psi-A-full}
\end{align}

The limiting weight-fluctuation variable is
\[
 \eta_B:=\frac1{v(0)}\sum_{i,j=1}^n v^{ij}(0)
 \left(
  \sqrt{V_0^i}X^i+\sqrt{V_0^j}X^j-2\xi
 \right).
\]
For each $i,j$,
\begin{align*}
 &\mathsf E\left[\xi\left(
  \sqrt{V_0^i}X^i+\sqrt{V_0^j}X^j-2\xi
 \right)\right]
 =\sum_{k=1}^n\left(
  \frac{v^{ik}(0)}{\Pi_0^i}
  +\frac{v^{jk}(0)}{\Pi_0^j}
 \right)-2v(0).
\end{align*}
It follows that
\begin{align}
 \mathsf E[\xi\eta_B]
 &=\sum_{i,j,k=1}^n\frac{v^{ij}(0)}{v(0)}
 \left(
  \frac{v^{ik}(0)}{\Pi_0^i}
  +\frac{v^{jk}(0)}{\Pi_0^j}
 \right)-2v(0)
 =\sqrt{v(0)} \psi_B.
 \label{psi-B-full}
\end{align}

Before applying \eqref{Fukasawa-expansion}, we replace its
deterministic volatility level by the one appearing in the
statement of the theorem. By \eqref{EV-two-scale} and
Karamata's theorem (noting that $r(\theta)+\sqrt\theta$ is
regularly varying with index $H$ if $0\leq H<1/2$, and with
index $1/2$ if $H\geq1/2$),
\[
    \widehat v(\theta)
    =
    \frac1\theta\int_0^\theta\mathsf E[V_t]\,\mathrm dt
    =
    \bar v(\theta)
    +o\bigl(r(\theta)+\sqrt\theta\bigr).
\]
Since $v(0)>0$, it follows that
\begin{equation}\label{deterministic-level-replacement}
    \sqrt{\widehat v(\theta)}
    =
    \sqrt{\bar v(\theta)}
    +o\bigl(r(\theta)+\sqrt\theta\bigr).
\end{equation}

We finally apply \eqref{Fukasawa-expansion}. Put
\[
 s(\theta):=r(\theta)+\sqrt\theta.
\]
Then $s$ is regularly varying at zero with index
\[
    \kappa:=H\wedge\frac12,
\]
and
\begin{align}
 \frac{V_\theta}{\mathsf E[V_\theta]}-1
 &=s(\theta)\left\{
  \frac{r(\theta)}{s(\theta)}\eta_{A,\theta}
  +\frac{\sqrt\theta}{s(\theta)}\eta_{B,\theta}
 \right\}+o_{L^1}(s(\theta)).
 \label{s-normalization}
\end{align}

Given any sequence $\theta_m\downarrow0$, select a subsequence
along which, after relabeling,
\[
 \alpha:=\lim_{m\to\infty}
 \frac{r(\theta_m)}{s(\theta_m)}
\]
exists. Along this subsequence, the limiting normalized
variance fluctuation is
\[
 \alpha\eta_A+(1-\alpha)\eta_B.
\]
The corresponding family in \eqref{s-normalization} is uniformly
integrable, and its joint limit with $\xi_\theta$ is centered
Gaussian. Hence \eqref{Fukasawa-expansion}, applied with the rate
$s$ and index $\kappa$, gives
\begin{align*}
 \sigma(\sqrt\theta z,\theta)
 &=\sqrt{\bar v(\theta)}
 +\frac{z}{\sqrt{v(0)}}
 \frac{\alpha\mathsf E[\xi\eta_A]
 +(1-\alpha)\mathsf E[\xi\eta_B]}
 {2\kappa+3}s(\theta)
 +o(s(\theta)),
\end{align*}
uniformly in $z$ on compact sets. The displayed expansion is
equivalent to the one in the theorem. Indeed, regular variation
implies
\[
\begin{array}{lll}
 H<1/2: & \alpha=1, & \kappa=H,\\
 H=1/2: & \alpha\in[0,1], & \kappa=1/2,\\
 H>1/2: & \alpha=0, & \kappa=1/2.
\end{array}
\]
Consequently, along the selected subsequence,
\begin{align*}
 &\frac{\alpha\mathsf E[\xi\eta_A]
 +(1-\alpha)\mathsf E[\xi\eta_B]}
 {2\kappa+3}s(\theta)
 \\
 &\qquad=
 \frac{\mathsf E[\xi\eta_A]}{2H+3}r(\theta)
 +\frac{\mathsf E[\xi\eta_B]}4\sqrt\theta
 +o(s(\theta)).
\end{align*}
Using \eqref{psi-A-full}, \eqref{psi-B-full}, and
\eqref{deterministic-level-replacement}, we obtain
\eqref{basket-IV-expansion} along this subsequence. Since every
sequence tending to zero has such a subsequence, the
subsequence principle yields \eqref{basket-IV-expansion} along
the full limit $\theta\downarrow0$.
\end{proof}

\section{Gaussian factor volatility models}
\label{sec:Gaussian-factor-models}

We now specialize Theorem~\ref{thm1} to volatility models driven
by Gaussian Volterra factors. Examples include multifactor
Bergomi models~\cite{Bergomi}, rough Bergomi
models~\cite{RoughBook}, and quintic volatility
models~\cite{AbiJaber}.

Let $(W^1,\dots,W^{d+n})$ be a $(d+n)$-dimensional standard
Brownian motion, and let $\{\mathscr F_t\}$ be its augmented
natural filtration. Suppose that
\begin{equation*}
    B_t^i=\sum_{l=1}^{d+n}\rho_l^iW_t^l,
    \qquad
    \sum_{l=1}^{d+n}|\rho_l^i|^2=1,
\end{equation*}
so that
\[
    \rho^{ij}
    =\sum_{l=1}^{d+n}\rho_l^i\rho_l^j.
\]
Define the residual return-covariance matrix
\begin{equation}\label{residual-covariance}
    C^\perp_{ij}
    :=\sum_{l=d+1}^{d+n}\rho_l^i\rho_l^j,
    \qquad i,j=1,\dots,n.
\end{equation}
Let $Z=(Z^1,\dots,Z^q)$ be the centered Gaussian Volterra
process defined by
\begin{equation*}
    Z_t^a
    =\sum_{l=1}^d\int_0^tK_{al}(t-s)\,\mathrm dW_s^l,
    \qquad a=1,\dots,q.
\end{equation*}
We assume that each kernel $K_{al}$ is locally square
integrable: for any $T>0$,
\begin{equation}\label{kernel-local-L2-bound}
\int_0^T
    K_{al}(t)^2 \, \mathrm{d}t < \infty
    \qquad a=1,\dots,q,\quad l=1,\dots,d.
\end{equation}
Assume also that $Z$ admits a continuous modification. Let
$r:(0,t_0)\to(0,\infty)$ be regularly varying at zero with
index $H\in(0,1/2]$. We assume that
\begin{equation}\label{kernel-short-time-assumption}
    \kappa_{al}
    :=\lim_{u\downarrow0}\frac{\sqrt u}{r(u)}K_{al}(u)
\end{equation}
exists and is finite for every $a,l$.

In this and the next sections,
we suppose that the constituent price processes $S^i$ are true
martingales with the constituent spot variance processes $V^i$ of the form
\begin{equation}\label{general-Gaussian-volatility}
    V_t^i=f^i(t,Z_t),
    \qquad i=1,\dots,n,
\end{equation}
where $f^i:[0,\infty)\times\mathbb R^q\to(0,\infty)$ is continuous
and continuously differentiable in the factor variables. 
Put
\[
    V_0^i:=f^i(0,0),
    \qquad
    \beta_a^i:=\partial_{z_a}\log f^i(0,0).
\]
and define
\begin{equation*}
    \gamma^{ki}
    :=
    \sum_{a=1}^q\sum_{l=1}^d
    \beta_a^i\rho_l^k\kappa_{al}.
\end{equation*}
Thus $\gamma^{ki}$ is the short-time exposure of the volatility
of constituent $i$ to the Brownian shock driving constituent
$k$.

\begin{theorem}\label{thm:Gaussian-factor-skew}
For every $i$,
assume that $(t,z)\mapsto\nabla_zf^i(t,z)$ is continuous at
$(0,0)$ and that there exist constants
$C,c>0$ such that
\begin{equation}\label{F-exponential-growth}
    f^i(t,z)
    +\frac1{f^i(t,z)}
    +\left|\nabla_z\log f^i(t,z)\right|
    \leq C\exp(c|z|),
    \qquad (t,z)\in[0,t_0]\times\mathbb R^q.
\end{equation}
  Then the assumptions of Theorem~\ref{thm1} hold with the rate $r$ and
\begin{equation}\label{psi-Gaussian-factor}
    \psi^{ki}
    =\frac{2}{2H+1}\gamma^{ki}.
\end{equation}
In particular,
\begin{equation}\label{psi-A-Gaussian-factor}
    \psi_A
    =
    \frac1{v(0)^{3/2}(2H+1)}
    \sum_{i,j,k=1}^n
    v^{ij}(0)\Pi_0^k\sqrt{V_0^k}
    \bigl(\gamma^{ki}+\gamma^{kj}\bigr),
\end{equation}
and if $v(0)>0$,
\begin{align}
    \sigma(\sqrt\theta z,\theta)
    &=\sqrt{\bar v(\theta)}
    +
    \frac{z r(\theta)}
         {v(0)^{3/2}(2H+1)(2H+3)}
    \sum_{i,j,k=1}^n
    v^{ij}(0)\Pi_0^k\sqrt{V_0^k}
    \bigl(\gamma^{ki}+\gamma^{kj}\bigr)
    \nonumber\\
    &\quad+
    \frac{\psi_B}{4}z\sqrt\theta
    +o\bigl(r(\theta)+\sqrt\theta\bigr),
    \label{Gaussian-factor-IV-expansion}
\end{align}
uniformly in $z$ on compact sets.
\end{theorem}

\begin{proof}
We first verify the delta-method expansion required in
Theorem~\ref{thm1}. The fundamental theorem of calculus gives
\begin{align*}
    f^i(\theta,Z_\theta)-f^i(\theta,0)
    &=\int_0^1\nabla_zf^i(\theta,uZ_\theta)
      \cdot Z_\theta\,\mathrm du\\
    &=\nabla_zf^i(0,0)\cdot Z_\theta+R_\theta^i,
\end{align*}
where
\[
    R_\theta^i
    :=\int_0^1
      \left\{\nabla_zf^i(\theta,uZ_\theta)
      -\nabla_zf^i(0,0)\right\}
      \cdot Z_\theta\,\mathrm du.
\]
The normalized Gaussian family
$\{Z_\theta/r(\theta)\}_{\theta\downarrow0}$ has uniformly
bounded moments of every order. Hence continuity at $(0,0)$,
\eqref{F-exponential-growth}, and the Gaussian exponential
moment bounds imply
\[
    R_\theta^i=o_{L^1}(r(\theta)).
\]
It follows that
\begin{equation}\label{uncentered-Gaussian-delta-expansion}
    V_\theta^i
    =f^i(\theta,0)
    +\sum_{a=1}^q\partial_{z_a}f^i(0,0)Z_\theta^a
    +o_{L^1}(r(\theta)).
\end{equation}
Since $\mathsf E[Z_\theta^a]=0$, taking expectations gives
\[
    \mathsf E[V_\theta^i]=f^i(\theta,0)+o(r(\theta)).
\]
Because $f^i(\theta,0)\to V_0^i>0$, subtracting expectations
in \eqref{uncentered-Gaussian-delta-expansion} and dividing by
$r(\theta)\mathsf E[V_\theta^i]$ yields
\begin{equation}\label{Gaussian-delta-expansion}
    \frac1{r(\theta)}
    \left(
        \frac{V_\theta^i}{\mathsf E[V_\theta^i]}-1
    \right)
    =\sum_{a=1}^q
      \beta_a^i\frac{Z_\theta^a}{r(\theta)}
      +o_{L^1}(1).
\end{equation}
The same exponential-moment estimate shows that the families
on the left-hand side are uniformly integrable.

By \eqref{kernel-short-time-assumption},
\[
    \left(
        \frac{Z_\theta^1}{r(\theta)},\dots,
        \frac{Z_\theta^q}{r(\theta)}
    \right)
\]
converges to a centered Gaussian vector $(Z^1,\dots,Z^q)$ with
covariance matrix
\begin{equation}\label{limiting-factor-covariance}
    \mathsf E[Z^aZ^b]
    =\frac1{2H}\sum_{l=1}^d\kappa_{al}\kappa_{bl}.
\end{equation}
It follows from \eqref{Gaussian-delta-expansion} that
\[
    Y^i=\sum_{a=1}^q\beta_a^iZ^a.
\]
Moreover,
\[
    X_\theta^k
    =\frac{B_\theta^k}{\sqrt\theta}+o_{\mathsf P}(1).
\]
Consequently, the joint limit of the $X_\theta^k$ and
$Y_\theta^i$ is centered Gaussian. Its cross-covariance is
\begin{align*}
    \mathsf E[X^kZ^a]
    &=\lim_{\theta\downarrow0}
      \frac1{r(\theta)\sqrt\theta}
      \sum_{l=1}^d\rho_l^k
      \int_0^\theta K_{al}(\theta-s)\,\mathrm ds
    \\
    &=\frac1{H+1/2}
      \sum_{l=1}^d\rho_l^k\kappa_{al}.
\end{align*}
Therefore,
\[
    \mathsf E[X^kY^i]
    =\frac1{H+1/2}
      \sum_{a=1}^q\sum_{l=1}^d
      \beta_a^i\rho_l^k\kappa_{al}
    =\frac{2}{2H+1}\gamma^{ki},
\]
which proves \eqref{psi-Gaussian-factor}. Substitution into the
definition of $\psi_A$ in Theorem~\ref{thm1} gives
\eqref{psi-A-Gaussian-factor}, and
\eqref{Gaussian-factor-IV-expansion} follows from
\eqref{basket-IV-expansion}.
\end{proof}

\begin{remark}[The boundary index $H=0$]
\upshape
The abstract Theorem~\ref{thm1} permits $H=0$. The present
Volterra criterion does not. Indeed, if $r$ is slowly varying,
$r(0+)=0$, and
\[
    K(u)\sim\kappa\frac{r(u)}{\sqrt u}
\]
with $\kappa>0$,
then
\[
    \mathsf E[Z_t^2]
    \asymp
    \int_0^t\frac{r(u)^2}{u}\,\mathrm du,
\]
which is generally larger than $r(t)^2$. Thus $Z_t/r(t)$ need
not be tight, and $r(0+)=0$ alone does not verify the
normalized-fluctuation hypothesis of Theorem~\ref{thm1}.
The case $H=0$ therefore requires a different normalization or
additional boundary assumptions and a separate analysis.
\end{remark}

The expansion \eqref{Gaussian-factor-IV-expansion} is uniform
for near-the-money strikes, but uniformity alone does not
justify differentiating it at the money. We therefore extend
the density argument of El Euch et al.~\cite{EFGR} to the basket. Although the basket is not
conditionally lognormal, the nondegeneracy of $C^\perp$ in
\eqref{residual-covariance} gives the uniform Schwartz bounds
that replace Lemma~3.4 of~\cite{EFGR}. The remaining part of
their characteristic-function and Fourier-inversion argument
then applies to the stochastic expansion of the log basket.

Put
\begin{equation}\label{sigma-zero-definition}
    \sigma_{0,\theta}^2
    :=\int_0^\theta\mathsf E[V_t]\,\mathrm dt,
    \qquad
    L_\theta
    :=\frac{\log(S_\theta/S_0)}{\sigma_{0,\theta}},
\end{equation}
and let $H_m$ denote the probabilists' Hermite polynomial of
degree $m$; in particular,
\[
    H_2(x)=x^2-1,
    \qquad
    H_3(x)=x^3-3x.
\]

\begin{theorem}[ATM basket skew and density expansion]
\label{thm:Gaussian-factor-ATM-skew-density}
Suppose that the assumptions of
Theorem~\ref{thm:Gaussian-factor-skew} hold, that $v(0)>0$,
and that $C^\perp$ in \eqref{residual-covariance} is positive
definite. Assume in addition that there exists $\delta>0$ such
that, for every $i$, $\nabla_z^2\log f^i(t,z)$ exists and is
continuous on $[0,\delta]\times\{z\in\mathbb R^q:|z|\leq\delta\}$.
Then $L_\theta$ in
\eqref{sigma-zero-definition} has
a smooth density $p_\theta$. If
\[
    s(\theta):=r(\theta)+\sqrt\theta,
    \qquad
    c_\theta
    :=\frac{\psi_A}{2H+3}r(\theta)
      +\frac{\psi_B}{4}\sqrt\theta,
\]
then, for every integer $m\geq0$,
\begin{equation}\label{basket-density-expansion}
 \sup_{x\in\mathbb R}(1+x^2)^m
 \left|p_\theta(x)-q_\theta(x)\right|
 =o\bigl(s(\theta)\bigr),
\end{equation}
where
\begin{align}
 q_\theta(x)
 &:={}
 \phi\left(x+\frac{\sigma_{0,\theta}}2\right)
 \left[1+c_\theta
 H_3\left(x+\frac{\sigma_{0,\theta}}2\right)
 \right].
 \label{basket-approximating-density}
\end{align}
Consequently,
\begin{equation}\label{basket-digital-expansion}
 \mathsf Q(S_\theta<S_0)
 -\Phi\left(\frac{\sigma_{0,\theta}}2\right)
 =\phi(0)c_\theta+o\bigl(s(\theta)\bigr).
\end{equation}
Equivalently,
\begin{equation}\label{ATM-digital-expansion}
 \mathsf Q(S_\theta<S_0)
 -\Phi\left(
   \frac{\sqrt\theta\,\sigma(0,\theta)}2
 \right)
 =\phi(0)c_\theta+o\bigl(s(\theta)\bigr).
\end{equation}
Moreover, $k\mapsto\sigma(k,\theta)$ is differentiable at
$k=0$ for all sufficiently small $\theta$, and
\begin{align}
 \partial_k\sigma(0,\theta)
 &=\frac{\psi_A}{2H+3}
   \frac{r(\theta)}{\sqrt\theta}
   +\frac{\psi_B}{4}
   +o\left(\frac{r(\theta)}{\sqrt\theta}+1\right).
 \label{Gaussian-factor-skew}
\end{align}
\end{theorem}

\begin{proof}
We divide the proof into three steps.

\medskip
\noindent
\textit{Step 1: Smoothness and uniform Fourier estimates.}

Let
\[
    \mathscr G_\theta
    :=
    \sigma\bigl(
        W_s^l:0\leq s\leq\theta,\ 1\leq l\leq d
    \bigr)\vee\mathscr N,
\]
where $\mathscr N$ denotes the null sets. Conditional on
$\mathscr G_\theta$, the vector
\[
    \ell_\theta
    :=
    \left(
        \log\frac{S_\theta^1}{S_0^1},
        \ldots,
        \log\frac{S_\theta^n}{S_0^n}
    \right)
\]
is Gaussian. Its conditional covariance matrix is
\begin{equation*}
    \Gamma_\theta
    =
    \int_0^\theta
    D_t C^\perp D_t\,\mathrm dt,
    \qquad
    D_t
    :=
    \operatorname{diag}
    \left(
        \sqrt{V_t^1},\ldots,\sqrt{V_t^n}
    \right).
\end{equation*}
Since $C^\perp$ is positive definite,
\begin{equation}\label{conditional-covariance-lower-bound}
    \Gamma_\theta
    \geq
    \lambda_{\min}(C^\perp)
    \left(
        \int_0^\theta
        \min_{1\leq i\leq n}V_t^i\,\mathrm dt
    \right)I_n.
\end{equation}

Moreover,
\begin{equation*}
    \log\frac{S_\theta}{S_0}
    =
    G(\ell_\theta),
    \qquad
    G(x)
    :=
    \log\left(
        \sum_{i=1}^n\Pi_0^i e^{x_i}
    \right).
\end{equation*}
The gradient of $G$ belongs to the unit simplex. In particular,
\[
    |\nabla G(x)|^2\geq\frac1n,
    \qquad x\in\mathbb R^n,
\]
and $G$ has no critical point. All derivatives of $G$ of
positive order are bounded.

The lower bound \eqref{conditional-covariance-lower-bound},
the reciprocal estimate in \eqref{F-exponential-growth}, and
the Gaussian exponential moment bounds permit repeated
conditional Gaussian integration by parts, with bounds uniform
for all sufficiently small $\theta$. Consequently,
$L_\theta$ has a smooth density $p_\theta$.
Repeated conditional
integration by parts also gives, for every pair of
nonnegative integers $a$ and $j$,
\begin{equation}\label{basket-Schwartz-Fourier-bound}
    \sup_{0<\theta<\theta_0}
    \int_{\mathbb R}
    |u|^j
    \left|
        \mathsf E\left[
            L_\theta^a e^{iuL_\theta}
        \right]
    \right|
    \mathrm du
    <\infty.
\end{equation}
This is the basket counterpart of the Fourier estimate in
Lemma~3.4 of~\cite{EFGR}. The conditional lognormality used
there is replaced here by the conditionally Gaussian vector
$\ell_\theta$, the nondegeneracy
\eqref{conditional-covariance-lower-bound}, and the fact that $G$ has no critical point.

\medskip
\noindent
\textit{Step 2: Stochastic and density expansions.}

For $i=1,\dots,n$, put
\[
    \boldsymbol\beta^i(t):=\nabla_z\log f^i(t,0).
\]
The local $C^2$ assumption, \eqref{F-exponential-growth}, and
the kernel asymptotics give, for every $p<\infty$,
$\|Z_t\|_{L^p}=O(r(t))$. Hence Taylor's theorem and Gaussian
tail estimates imply
\[
    \sqrt{V_t^i}
    =\sqrt{m^i(t)}
     \left\{1+\frac12\boldsymbol\beta^i(t)\cdot Z_t\right\}
     +O_{L^p}\bigl(r(t)^2\bigr)
\]
and
\[
    V_t^i
    =m^i(t)\left\{1+\boldsymbol\beta^i(t)\cdot Z_t\right\}
     +O_{L^p}\bigl(r(t)^2\bigr).
\]
Indeed, Taylor's theorem gives these estimates on a fixed
neighborhood of the origin, while the complement is negligible
to every polynomial order by the Gaussian tail estimate and
\eqref{F-exponential-growth}. Notice that the coefficients are
retained at $(t,0)$; no rate of convergence of
$\boldsymbol\beta^i(t)$ to $\boldsymbol\beta^i(0)$ is needed.

Put
\[
    v_{\mathrm{fr}}(t)
    :=
    \sum_{i,j=1}^n
    \Pi_0^i\Pi_0^j\rho^{ij}
    \sqrt{m^i(t)m^j(t)}.
\]
Since $v_{\mathrm{fr}}(t)\to v(0)>0$, it is positive for all
sufficiently small $t$. For $l=1,\dots,d+n$, define the
deterministic integrand
\begin{equation}\label{leading-Gaussian-integrand}
    h_l(t)
    :=
    \sqrt{\frac{\mathsf E[V_t]}{v_{\mathrm{fr}}(t)}}
    \sum_{i=1}^n
    \Pi_0^i\sqrt{m^i(t)}\rho_l^i.
\end{equation}
Then
\[
    \sum_{l=1}^{d+n}h_l(t)^2
    =\mathsf E[V_t],
\]
and hence
\begin{equation}\label{leading-standard-normal}
    M_\theta^{(0)}
    :=
    \frac1{\sigma_{0,\theta}}
    \sum_{l=1}^{d+n}
    \int_0^\theta h_l(t)\,\mathrm dW_t^l
\end{equation}
is standard normal for every sufficiently small $\theta$.
The same product and weight estimates as in the proof of
Theorem~\ref{thm1}, now using the preceding second-order
expansions, give
\[
    \mathsf E[V_t]
    =v_{\mathrm{fr}}(t)
     +O\bigl(r(t)^2+r(t)\sqrt t+t\bigr)
    =v_{\mathrm{fr}}(t)+O\bigl(s(t)^2\bigr).
\]
Consequently, the square-root factor in
\eqref{leading-Gaussian-integrand} equals $1+O(s(t)^2)$.
Replacing the frozen-weight Gaussian integral by
\eqref{leading-standard-normal} therefore changes the
normalized log return by $O_{L^p}(s(\theta)^2)$ for every
$p<\infty$; this difference is included in the remainder
below.

The stochastic expansions established in the proof of
Theorem~\ref{thm1}, together with the second-order Taylor
expansion
\[
    G(x)
    =
    \Pi_0^\top x
    +
    \frac12x^\top
    \left\{
        \operatorname{diag}(\Pi_0)
        -
        \Pi_0\Pi_0^\top
    \right\}x
    +
    O(|x|^3),
\]
give
\begin{equation}\label{two-scale-log-return-expansion}
    L_\theta
    =
    M_\theta^{(0)}
    -
    \frac{\sigma_{0,\theta}}2
    +
    D_\theta
    +
    R_\theta,
\end{equation}
where $D_\theta$ is obtained by retaining the terms linear in
$Z_t$ in the volatility integrands, with the deterministic
coefficients $\boldsymbol\beta^i(t)$, and the centered quadratic
term in the expansion of $G$. Thus $D_\theta$ is a centered
element of the second Wiener chaos and, for every $p<\infty$,
\begin{equation}\label{log-return-remainder-estimates}
    \|D_\theta\|_{L^p}
    =
    O\bigl(s(\theta)\bigr),
    \qquad
    \|R_\theta\|_{L^p}
    =
    O\bigl(s(\theta)^2\bigr).
\end{equation}

We use the following elementary projection identity. If
$M=I_1(h)$ is standard normal and $D=I_2(f)$ belongs to the
second Wiener chaos of the same isonormal Gaussian process,
then
\begin{equation}\label{second-chaos-projection-identity}
    \mathsf E[D\mid M]
    =
    \frac12\mathsf E[D H_2(M)]H_2(M).
\end{equation}
Indeed, the conditional expectation is the orthogonal
projection of $D$ onto the second Wiener chaos generated by
$M$, which is spanned by $H_2(M)$, and
\[
    \mathsf E[H_2(M)^2]=2.
\]

Set
\begin{equation*}
    b_\theta
    :=
    \frac12
    \mathsf E\left[
        D_\theta H_2(M_\theta^{(0)})
    \right].
\end{equation*}
By \eqref{log-return-remainder-estimates},
$b_\theta=O(s(\theta))$, and
\eqref{second-chaos-projection-identity} gives
\begin{equation}\label{conditional-perturbation-projection}
    \mathsf E\left[
        D_\theta\mid M_\theta^{(0)}
    \right]
    =
    b_\theta H_2(M_\theta^{(0)}).
\end{equation}

Put
\[
    a_\theta:=\frac{\sigma_{0,\theta}}2.
\]
Define
\[
    \bar q_\theta(x)
    :=
    \phi(x+a_\theta)
    \left\{
        1+b_\theta H_3(x+a_\theta)
    \right\}.
\]
We now follow the Fourier argument of El Euch et
al.~\cite{EFGR}. For every nonnegative integer $a$, Taylor's
formula, \eqref{log-return-remainder-estimates}, and
\eqref{conditional-perturbation-projection} give, uniformly in
$u\in\mathbb R$,
\[
 \left|
  \mathsf E[L_\theta^a e^{iuL_\theta}]
  -\int_{\mathbb R}x^a e^{iux}\bar q_\theta(x)\,\mathrm dx
 \right|
 \leq C_a(1+|u|)^2s(\theta)^2.
\]
Indeed, the first-order term is identified by integration by
parts using
\[
    \frac{\mathrm d}{\mathrm dx}
    \{H_2(x)\phi(x)\}=-H_3(x)\phi(x),
\]
and all second-order terms are $O((1+|u|)^2s(\theta)^2)$.
Choose $U_\theta=s(\theta)^{-\varepsilon}$ with
$0<\varepsilon<1/3$. Integration over $|u|\leq U_\theta$ gives
$O(s(\theta)^{2-3\varepsilon})=o(s(\theta))$. For
$|u|>U_\theta$, \eqref{basket-Schwartz-Fourier-bound}, together
with the corresponding immediate bound for $\bar q_\theta$,
gives $o(s(\theta))$ by taking a sufficiently high power of
$|u|$. Fourier inversion, applied with
$a=0,2,\dots,2m$, therefore gives, for every integer $m\geq0$,
\begin{equation}\label{density-expansion-with-b}
    \sup_{x\in\mathbb R}(1+x^2)^m
    \left|
        p_\theta(x)-\bar q_\theta(x)
    \right|
    =
    o\bigl(s(\theta)\bigr).
\end{equation}
It remains only to identify $b_\theta$. The standard
Black--Scholes inversion of
\eqref{density-expansion-with-b}, carried out exactly as in
the proof of Theorem~2.3 of~\cite{EFGR}, gives
\begin{equation*}
    \sigma(\sqrt\theta z,\theta)
    =
    \sqrt{\widehat v(\theta)}
    +
    b_\theta z
    +
    o\bigl(s(\theta)\bigr),
\end{equation*}
uniformly in $z$ on compact sets. On the other hand,
Theorem~\ref{thm1} and
\eqref{deterministic-level-replacement} give
\[
    \sigma(\sqrt\theta z,\theta)
    =
    \sqrt{\widehat v(\theta)}
    +
    c_\theta z
    +
    o\bigl(s(\theta)\bigr),
\]
where
\[
    c_\theta
    =
    \frac{\psi_A}{2H+3}r(\theta)
    +
    \frac{\psi_B}{4}\sqrt\theta.
\]
Comparison at any fixed $z\neq0$ gives
\begin{equation}\label{Hermite-coefficient-identification}
    b_\theta
    =
    c_\theta+o\bigl(s(\theta)\bigr).
\end{equation}
Substituting
\eqref{Hermite-coefficient-identification} into
\eqref{density-expansion-with-b} proves
\eqref{basket-density-expansion}.

\medskip
\noindent
\textit{Step 3: Digital probability and ATM skew.}

Integrating \eqref{basket-density-expansion} over
$(-\infty,0)$ gives
\begin{align*}
    \mathsf Q(S_\theta<S_0)
    &=
    \Phi(a_\theta)
    +
    c_\theta
    \int_{-\infty}^{a_\theta}
        H_3(y)\phi(y)\,\mathrm dy
    +
    o\bigl(s(\theta)\bigr)\\
    &=
    \Phi(a_\theta)
    +
    c_\theta
    (1-a_\theta^2)\phi(a_\theta)
    +
    o\bigl(s(\theta)\bigr)\\
    &=
    \Phi(a_\theta)
    +
    \phi(0)c_\theta
    +
    o\bigl(s(\theta)\bigr),
\end{align*}
because $a_\theta=O(\sqrt\theta)$ and
$c_\theta=O(s(\theta))$. This proves
\eqref{basket-digital-expansion}.

Furthermore,
\[
    \frac{\sigma_{0,\theta}}{\sqrt\theta}
    =
    \sqrt{\widehat v(\theta)}.
\]
By \eqref{deterministic-level-replacement} and
\eqref{basket-IV-expansion} at $z=0$,
\[
    \sqrt\theta\,\sigma(0,\theta)
    -
    \sigma_{0,\theta}
    =
    o\bigl(\sqrt\theta\,s(\theta)\bigr).
\]
Consequently, $\Phi(a_\theta)$ in
\eqref{basket-digital-expansion} may be replaced by
\[
    \Phi\left(
        \frac{\sqrt\theta\,\sigma(0,\theta)}2
    \right),
\]
which proves \eqref{ATM-digital-expansion}.

The smooth density established in Step~1 implies that the put
price is differentiable in log-strike. Since the
Black--Scholes vega is strictly positive, the implicit-function
theorem shows that $k\mapsto\sigma(k,\theta)$ is differentiable
near $k=0$. Differentiating the implied-volatility pricing
identity at the money gives the exact relation
\begin{align}
    \mathsf Q(S_\theta<S_0)
    -
    \Phi\left(
        \frac{\sqrt\theta\,\sigma(0,\theta)}2
    \right)
    &=
    \phi\left(
        \frac{\sqrt\theta\,\sigma(0,\theta)}2
    \right)
    \sqrt\theta\,\partial_k\sigma(0,\theta).
    \label{digital-skew-identity}
\end{align}
Combining this identity with
\eqref{ATM-digital-expansion} and using
$\sqrt\theta\,\sigma(0,\theta)\to0$, we obtain
\[
    \partial_k\sigma(0,\theta)
    =
    \frac{c_\theta}{\sqrt\theta}
    +
    o\left(
        \frac{s(\theta)}{\sqrt\theta}
    \right).
\]
Therefore,
\[
    \partial_k\sigma(0,\theta)
    =
    \frac{\psi_A}{2H+3}
    \frac{r(\theta)}{\sqrt\theta}
    +
    \frac{\psi_B}{4}
    +
    o\left(
        \frac{r(\theta)}{\sqrt\theta}+1
    \right).
\]
This proves \eqref{Gaussian-factor-skew}.
\end{proof}

\section{Skew stickiness ratio}
\label{sec:stickiness}

We now study the short-maturity skew stickiness ratio (SSR) for the
basket under the same Gaussian factor model as in the previous section.
For the proof, it is convenient to use the alternative
total-implied-variance representation in~\cite{F26}.

Fix a maturity $T \in (0,t_0)$. For $t<T$ and $K>0$, let
\[
    P_t(K)
    :=
    \mathsf E\bigl[(K-S_T)_+\mid\mathscr F_t\bigr]
\]
and define the total implied variance $\Sigma_t(K)$ by
\[
    P_t(K)
    =
    p\bigl(S_t,K,\Sigma_t(K)\bigr),
\]
where $p(s,K,w)$ denotes the Black--Scholes put price with spot
$s$, strike $K$, and total variance $w$. Put
\[
    \Sigma_t^S:=\Sigma_t(S_t)
\]
and write
\[
    \Sigma_t'(K):=\partial_K\Sigma_t(K).
\]
Equivalently,
\[
    \Sigma_t(K)
    =
    (T-t)\,
    \sigma_t\!\left(\log\frac{K}{S_t},T-t\right)^2,
\]
and hence
\[
    \Sigma_t^S
    =
    (T-t)\sigma_t(0,T-t)^2.
\]

The skew stickiness ratio is originally defined by
\begin{equation}\label{SSR-original-definition}
    \operatorname{SSR}_t(T)
    :=
    \frac{1}{\partial_k\sigma_t(0,T-t)}
    \frac{
        \mathrm d\langle\sigma^S,\log S\rangle_t
    }{
        \mathrm d\langle\log S\rangle_t
    },
    \qquad
    \sigma_t^S:=\sigma_t(0,T-t),
\end{equation}
whenever the denominator is nonzero.

The following equivalent representation is more convenient for
our purposes.

\begin{lemma}[Total-variance representation]
\label{lem:Fukasawa-total-variance}
Suppose that $\sigma^S$, $\Sigma^S$, and $S$ are continuous
semimartingales and that $\Sigma_t(K)$ is differentiable with
respect to $K$ at $K=S_t$. Then, whenever
$\Sigma_t'(S_t)\neq0$,
\begin{equation}\label{SSR-total-variance-representation}
    \operatorname{SSR}_t(T)
    =
    \frac{1}{\Sigma_t'(S_t)}
    \frac{
        \mathrm d\langle\Sigma^S,S\rangle_t
    }{
        \mathrm d\langle S\rangle_t
    }.
\end{equation}
\end{lemma}

\begin{proof}
Differentiating
\[
    \Sigma_t(K)
    =
    (T-t)\,
    \sigma_t\!\left(\log\frac{K}{S_t},T-t\right)^2
\]
with respect to $K$ and evaluating at $K=S_t$ gives
\begin{equation*}
    \Sigma_t'(S_t)
    =
    \frac{2(T-t)\sigma_t^S}{S_t}
    \partial_k\sigma_t(0,T-t).
\end{equation*}
Moreover, since
\[
    \Sigma_t^S=(T-t)(\sigma_t^S)^2,
\]
we have
\[
    \frac{
        \mathrm d\langle\Sigma^S,S\rangle_t
    }{
        \mathrm d\langle S\rangle_t
    }
    =
    2(T-t)\sigma_t^S
    \frac{
        \mathrm d\langle\sigma^S,S\rangle_t
    }{
        \mathrm d\langle S\rangle_t
    }.
\]
On the other hand,
\[
    \frac{
        \mathrm d\langle\sigma^S,\log S\rangle_t
    }{
        \mathrm d\langle\log S\rangle_t
    }
    =
    S_t
    \frac{
        \mathrm d\langle\sigma^S,S\rangle_t
    }{
        \mathrm d\langle S\rangle_t
    }.
\]
Combining these identities proves
\eqref{SSR-total-variance-representation}.
\end{proof}

We next verify explicitly the Malliavin differentiability
required by Theorem~1 of~\cite{F26}. Let
$\mathcal D^l=\{\mathcal D_u^l\}_{u\geq0}$ denote the
Malliavin--Shigekawa derivative with respect to $W^l$. On a
smooth cylindrical random variable
\[
    G
    =f\bigl(W_{t_1},\ldots,W_{t_m}\bigr),
    \qquad f\in C_b^\infty(\mathbb R^{m(d+n)}),
\]
it is defined by
\[
    \mathcal D_u^lG
    =
    \sum_{r=1}^m
    \partial_{(r,l)}f
       \bigl(W_{t_1},\ldots,W_{t_m}\bigr)
    1_{[0,t_r]}(u).
\]
This densely defined operator is closable as an operator from
$L^2(\Omega)$ to
$L^2(\Omega\times\mathbb R_+;\mathbb R^{d+n})$. We denote its
closure by
\[
    \mathcal D=(\mathcal D^1,\ldots,\mathcal D^{d+n})
\]
and its domain by $\mathbb D^{1,2}$.

For notational convenience, set
\[
    K_{al}=0,
    \qquad l=d+1,\dots,d+n.
\]
Since the factors are Gaussian Volterra processes,
\[
    \mathcal D_u^lZ_s^a
    =K_{al}(s-u)1_{\{u <  s\}}.
\]
Consequently, the Malliavin chain rule gives
\begin{equation}\label{Malliavin-log-V}
    \mathcal D_t^l\log V_s^j
    =\Lambda_l^j(s,t)1_{\{t<  s\}},
\end{equation}
where
\begin{equation*}
    \Lambda_l^j(s,t)
    :=
    \sum_{a=1}^q
    \partial_{z_a}\log f^j(s,Z_s)K_{al}(s-t),
    \qquad t <  s.
\end{equation*}
In particular,
\[
    \mathcal D_t^l\sqrt{V_s^j}
    =\frac12\sqrt{V_s^j}\Lambda_l^j(s,t)1_{\{t < s\}},
    \qquad
    \mathcal D_t^lV_s^j
    =V_s^j\Lambda_l^j(s,t)1_{\{t <  s\}}.
\]

For each constituent,
\begin{align}
    \mathcal D_t^l\log S_T^j
    &=\rho_l^j\sqrt{V_t^j}
      +\frac12\int_t^T
        \sqrt{V_s^j}\Lambda_l^j(s,t)\,\mathrm dB_s^j
    -\frac12\int_t^T
        V_s^j\Lambda_l^j(s,t)\,\mathrm ds.
    \label{explicit-D-log-Sj}
\end{align}
The Gaussian exponential moment bounds, together with
\eqref{F-exponential-growth} and
\eqref{kernel-local-L2-bound}, imply
\[
    \mathsf E\left[
      \sum_{l=1}^{d+n}\int_0^T
      \left|\mathcal D_t^l\log S_T^j\right|^2\,\mathrm dt
    \right]<\infty.
\]
Thus $\log S_T^j\in\mathbb D^{1,2}$. Since
\begin{equation}\label{explicit-D-log-basket}
    \mathcal D_t^l\log S_T
    =\sum_{j=1}^n\Pi_T^j
      \mathcal D_t^l\log S_T^j
\end{equation}
and $0\leq\Pi_T^j\leq1$, it follows that
$\log S_T\in\mathbb D^{1,2}$.

\begin{equation*}
    \Lambda^{ij}(s,t)
    :=
    \sum_{l=1}^d\rho_l^i\Lambda_l^j(s,t).
\end{equation*}
It follows from \eqref{kernel-short-time-assumption}, the
continuity of $\nabla_z\log f^j$, and
\eqref{F-exponential-growth} that, for every $p<\infty$,
\begin{equation}\label{Lambda-short-time-limit}
    \frac{\sqrt s}{r(s)}\Lambda^{ij}(s,0)
    \longrightarrow
    \gamma^{ij}
    :=\sum_{a=1}^q\sum_{l=1}^d
      \rho_l^i\beta_a^j\kappa_{al}
\end{equation}
in $L^p$ as $s\downarrow0$.

\begin{lemma}[Option-surface regularity]
\label{lem:option-surface-regularity}
Suppose that $C^\perp$ in \eqref{residual-covariance} is
positive definite. Then, for every $0\leq t<T$, the conditional
law of $S_T$ given $\mathscr F_t$ has a continuous density on
$(0,\infty)$. In particular, $K\mapsto P_t(K)$ is twice
continuously differentiable.

More precisely, if $T_0<T$ and $I\Subset(0,\infty)$, the fields
$P_t(K)$ and $\partial_KP_t(K)$ admit martingale
representations
\begin{align*}
    P_t(K)&=P_0(K)+\sum_{l=1}^{d+n}
      \int_0^t f_s^l(K)\,\mathrm dW_s^l,\\
    \partial_KP_t(K)&=\partial_KP_0(K)+
      \sum_{l=1}^{d+n}\int_0^t g_s^l(K)\,\mathrm dW_s^l,
\end{align*}
whose integrands have versions continuous in $K\in I$ and
satisfy
\[
    \mathsf E\left[\int_0^{T_0}
      \sup_{K\in I}\sum_{l=1}^{d+n}
      \bigl(|f_s^l(K)|^2+|g_s^l(K)|^2\bigr)\,\mathrm ds
    \right]<\infty.
\]
Consequently, the option-price field satisfies the regularity
hypotheses of Theorem~1 of~\cite{F26}; the
implied total-variance field $\Sigma_t(K)$ and its ATM
evaluation $\Sigma_t^S$ have the semimartingale and
strike-regularity properties used in
Lemma~\ref{lem:Fukasawa-total-variance}.
\end{lemma}

\begin{proof}
Condition on $\mathscr F_t$ and on the paths of the factor
Brownian motions $(W^1,\dots,W^d)$ over $[t,T]$. The future
volatility paths are then fixed, while the Brownian motions
$(W^{d+1},\dots,W^{d+n})$ remain independent. The conditional
law of
$(\log S_T^1,\dots,\log S_T^n)$ is Gaussian with covariance
matrix
\[
    \Gamma_{t,T}^{ij}
    =C_{ij}^\perp\int_t^T
      \sqrt{V_s^iV_s^j}\,\mathrm ds.
\]
For $x\in\mathbb R^n\setminus\{0\}$,
\begin{align*}
    x^\top\Gamma_{t,T}x
    &=\int_t^T
      \bigl(\operatorname{diag}(\sqrt{V_s^1},\dots,
      \sqrt{V_s^n})x\bigr)^\top
      C^\perp
      \bigl(\operatorname{diag}(\sqrt{V_s^1},\dots,
      \sqrt{V_s^n})x\bigr)\,\mathrm ds
    >0,
\end{align*}
because $C^\perp$ is positive definite and $V_s^i>0$.
Therefore the conditional vector of constituent log prices has
a smooth density. Moreover, for $T_0<T$,
\[
    \lambda_{\min}(\Gamma_{t,T})
    \geq
    \lambda_{\min}(C^\perp)(T-t)
    \min_{1\leq i\leq n}\inf_{t\leq s\leq T}V_s^i,
    \qquad t\leq T_0.
\]
The reciprocal bound in \eqref{F-exponential-growth}, the
continuity of $Z$, and Fernique's theorem imply that the
right-hand side has inverse moments of every order, locally
uniformly in $t<T$. Hence the conditional Gaussian density and
all the derivatives needed below have integrable bounds,
locally uniformly in $(t,K)\in[0,T)\times(0,\infty)$.

The image of the conditional Gaussian vector under
\[
    (x_1,\dots,x_n)\longmapsto
    \sum_{i=1}^ne^{x_i}
\]
has a continuous density on $(0,\infty)$: the map is a smooth
submersion, and the preceding bounds justify the coarea formula
and integration over the future factor paths. Differentiation
under the conditional expectation therefore gives
\[
    \partial_KP_t(K)
    =\mathsf Q(S_T<K\mid\mathscr F_t),
    \qquad
    \partial_{KK}P_t(K)
    =p_{S_T\mid\mathscr F_t}(K),
\]
with the asserted joint continuity.

The Brownian martingale-representation theorem applied to
$P(K)$ and $\partial_KP(K)$ gives the displayed integrands.
Conditional Gaussian differentiation, together with the
inverse-covariance moment bounds above, gives versions that are
continuous in $K$ and the stated local square-integrability
bound. These are precisely the option-surface regularity
conditions used in~\cite{F26}. Finally, the implicit-function
theorem applies to the Black--Scholes pricing identity because
its derivative with respect to total variance is strictly
positive. The random-field It\^o--Wentzell formula then gives
the asserted semimartingale properties of $\Sigma_t(K)$ and
$\Sigma_t^S$.
\end{proof}

Differentiating the option-pricing identity with respect to
strike gives
\begin{equation}\label{Sigma-prime-probability}
    \Sigma_t'(S_t)
    =
    \frac{2\sqrt{\Sigma_t^S}}
         {S_t\phi(\sqrt{\Sigma_t^S}/2)}
    Y_t^T,
\end{equation}
where
\begin{equation}\label{Y-definition}
    Y_t^T
    :=
    \mathsf Q(S_T<S_t\mid\mathscr F_t)
    -
    \Phi\left(\frac{\sqrt{\Sigma_t^S}}{2}\right).
\end{equation}
Here $\Phi$ and $\phi$ denote the standard normal distribution
function and density, respectively.

Since $\log S_T\in\mathbb D^{1,2}$, the same
Clark--Ocone argument as in Theorem~1 of~\cite{F26}, applied
to the $d+n$ Brownian drivers, gives
\begin{equation}\label{stickiness-Malliavin-representation}
    \frac{
        \mathrm d\langle\Sigma^S,S\rangle_t
    }{
        \mathrm d\langle S\rangle_t
    }
    =
    \frac{2\sqrt{\Sigma_t^S}}
         {S_t\phi(\sqrt{\Sigma_t^S}/2)}
    A_t^T,
\end{equation}
where
\begin{align}
    A_t^T
    &:=
    \mathsf E\left[
      1_{\{S_T<S_t\}}\frac{S_T}{S_t}
      \left(
        1
        -
        \frac1{V_t}
        \sum_{l=1}^{d+n}\sum_{i=1}^n
        \Pi_t^i\sqrt{V_t^i}\rho_l^i
        \mathcal D_t^l\log S_T
      \right)
      \,\middle|\,\mathscr F_t
    \right].
    \label{A-definition}
\end{align}
Indeed,
\[
    \frac{
        \mathrm d\langle W^l,S\rangle_t
    }{
        \mathrm d\langle S\rangle_t
    }
    =
    \frac1{S_tV_t}
    \sum_{i=1}^n
    \Pi_t^i\sqrt{V_t^i}\rho_l^i.
\]

Combining
\eqref{SSR-total-variance-representation},
\eqref{Sigma-prime-probability}, and
\eqref{stickiness-Malliavin-representation} yields the exact
identity
\begin{equation}\label{SSR-A-over-Y}
    \operatorname{SSR}_t(T)
    =
    \frac{A_t^T}{Y_t^T},
\end{equation}
whenever $Y_t^T\neq0$, or equivalently
$\Sigma_t'(S_t)\neq0$. This is the total-variance
representation that we use in the short-maturity analysis.

\begin{theorem}\label{thm:SSR-limit}
Suppose that the assumptions of
Theorem~\ref{thm:Gaussian-factor-skew} hold, that $v(0)>0$,
and that $C^\perp$ in \eqref{residual-covariance} is positive
definite. Assume in addition that there exists $\delta>0$ such
that, for every $i$, $\nabla_z^2\log f^i(t,z)$ exists and is
continuous on $[0,\delta]\times\{z\in\mathbb R^q:|z|\leq\delta\}$.
 Then, as
$\theta\downarrow0$,
\begin{equation}
    A_0^\theta
    =
    \frac{\phi(0)}2
    \left(
        \psi_A r(\theta)+\psi_B\sqrt\theta
    \right)
    +
    o\bigl(r(\theta)+\sqrt\theta\bigr).
    \label{A-total-expansion}
\end{equation}
\begin{equation}
    Y_0^\theta
    =
    \phi(0)
    \left(
        \frac{\psi_A}{2H+3}r(\theta)
        +
        \frac{\psi_B}{4}\sqrt\theta
    \right)
    +
    o\bigl(r(\theta)+\sqrt\theta\bigr).
    \label{Y-total-expansion}
\end{equation}
Consequently, if
\[
    \psi_A\neq0
    \quad\text{when }H<\frac12,
\]
and, when $H=1/2$, if
\[
    \liminf_{\theta\downarrow0}
    \frac{|\psi_A r(\theta)+\psi_B\sqrt\theta|}
         {r(\theta)+\sqrt\theta}>0,
\]
then
\begin{equation*}
    \lim_{\theta\downarrow0}
    \operatorname{SSR}_0(\theta)
    =
    H+\frac32.
\end{equation*}
In particular, the limit is equal to $2$ when $H=1/2$.
\end{theorem}

\begin{proof}
We first derive the expansion of $A_0^\theta$. By
\eqref{explicit-D-log-Sj} and
\eqref{explicit-D-log-basket},
\begin{align}
    \mathcal D_t^l\log S_T
    &=\sum_{j=1}^n\Pi_T^j\sqrt{V_t^j}\rho_l^j
    \nonumber\\
    &\quad+\frac12\sum_{j=1}^n\Pi_T^j\left(
          \int_t^T
          \sqrt{V_s^j}\Lambda_l^j(s,t)\,\mathrm dB_s^j
          -\int_t^T
          V_s^j\Lambda_l^j(s,t)\,\mathrm ds
        \right).
        \nonumber
\end{align}
Substitution into \eqref{A-definition}, with $t=0$ and
$T=\theta$, gives
\[
    A_0^\theta
    =
    A_{A,0}^\theta+A_{B,0}^\theta,
\]
where
\begin{align}
    A_{A,0}^\theta
    &:=
    -\mathsf E\Bigg[
      \frac{1_{\{S_\theta<S_0\}}S_\theta}{2S_0V_0}
      \sum_{i,j=1}^n
      \Pi_\theta^j\Pi_0^i\sqrt{V_0^i}
      \bigg(
        \int_0^\theta
        \sqrt{V_s^j}\Lambda^{ij}(s,0)\,\mathrm dB_s^j
    \nonumber\\[-1mm]
    &\hspace{5.3cm}
        -
        \int_0^\theta
        V_s^j\Lambda^{ij}(s,0)\,\mathrm ds
      \bigg)
    \Bigg],
    \label{AA-definition}\\
    A_{B,0}^\theta
    &:=
    \mathsf E\left[
      1_{\{S_\theta<S_0\}}\frac{S_\theta}{S_0}
      \left(
        1-\frac1{V_0}
        \sum_{i,j=1}^n
        \Pi_\theta^j\Pi_0^i
        \sqrt{V_0^iV_0^j}\rho^{ij}
      \right)
    \right].
    \nonumber
\end{align}

For the first term, define
\[
    \widehat X_\theta^j
    :=
    \frac1{r(\theta)}
    \int_0^\theta\frac{r(s)}{\sqrt s}\,\mathrm dB_s^j.
\]
Jointly with $\xi_\theta$, the vector
\[
    \bigl(
      \widehat X_\theta^1,\ldots,\widehat X_\theta^n,
      \xi_\theta
    \bigr)
\]
converges to a centered Gaussian vector
\[
    \bigl(
      \widehat X^1,\ldots,\widehat X^n,\xi
    \bigr)
\]
satisfying
\begin{equation}\label{hatX-xi-covariance}
    \mathsf E[\widehat X^j\xi]
    =
    \frac1{H+1/2}
    \sum_{k=1}^n
    \Pi_0^k\sqrt{V_0^k}\rho^{kj}.
\end{equation}
Indeed, Karamata's theorem gives
\[
    \mathsf E[(\widehat X_\theta^j)^2]
    =\frac1{r(\theta)^2}
      \int_0^\theta\frac{r(s)^2}{s}\,\mathrm ds
    \longrightarrow\frac1{2H},
\]
and the same theorem gives
\eqref{hatX-xi-covariance}.
Since $\mathsf E[\xi^2]=v(0)$, Gaussian regression gives
\begin{equation*}
    \mathsf E[
      \widehat X^j1_{\{\xi<0\}}
    ]
    =
    -\phi(0)
    \frac{
        \mathsf E[\widehat X^j\xi]
    }{
        \sqrt{v(0)}
    }.
\end{equation*}

By \eqref{Lambda-short-time-limit}, Karamata's theorem, and the
Gaussian moment bounds,
\begin{align*}
    \frac1{r(\theta)}
    \int_0^\theta
      \sqrt{V_s^j}\Lambda^{ij}(s,0)\,\mathrm dB_s^j
    &\Longrightarrow
      \sqrt{V_0^j}\gamma^{ij}\widehat X^j,
\end{align*}
jointly with $\xi_\theta$, whereas
\[
    \mathsf E\left|
      \int_0^\theta V_s^j\Lambda^{ij}(s,0)\,\mathrm ds
    \right|
    =O\bigl(r(\theta)\sqrt\theta\bigr)
    =o(r(\theta)).
\]
Furthermore, \eqref{return-xi-relation} gives
$1_{\{S_\theta<S_0\}}=1_{\{\xi_\theta<0\}}$.
Since $\mathsf Q(\xi=0)=0$, the continuous mapping theorem
implies that this indicator and the normalized stochastic
integrals above converge jointly in law to
$1_{\{\xi<0\}}$ and their respective limits.
After localization on compact sets of factor and price paths,
the normalized products in \eqref{AA-definition} are uniformly
integrable. The localization can be removed using the Gaussian
moment bounds and the true-martingale property of the prices.
Since $S_\theta/S_0\to1$ and
$\Pi_\theta^j\to\Pi_0^j$ in probability, Slutsky's theorem
and uniform integrability therefore allow us to pass to
expectations and obtain
\begin{equation*}
    \frac{A_{A,0}^\theta}{r(\theta)}
    \longrightarrow
    \frac{\phi(0)}
         {v(0)^{3/2}(2H+1)}
    \sum_{i,j,k=1}^n
    v^{jk}(0)\Pi_0^i\sqrt{V_0^i}\gamma^{ij}.
\end{equation*}
By the symmetry of $v^{jk}(0)$,
\begin{align*}
    &\sum_{i,j,k=1}^n
    v^{jk}(0)\Pi_0^i\sqrt{V_0^i}\gamma^{ij}
    =
    \frac12
    \sum_{i,j,k=1}^n
    v^{ij}(0)\Pi_0^k\sqrt{V_0^k}
    \bigl(\gamma^{ki}+\gamma^{kj}\bigr).
\end{align*}
It follows from \eqref{psi-A-Gaussian-factor} that
\begin{equation}\label{AA-limit-psiA}
    A_{A,0}^\theta
    =
    \frac{\phi(0)}2
    \psi_A r(\theta)
    +
    o(r(\theta)).
\end{equation}

We next consider $A_{B,0}^\theta$. Let
\[
    U_\theta^j
    :=
    \frac1{\sqrt\theta}
    \left(
        \frac{\Pi_\theta^j}{\Pi_0^j}-1
    \right).
\]
The first-order weight expansion gives, jointly with
$\xi_\theta$,
\begin{equation*}
    U_\theta^j
    \Longrightarrow
    U^j
    :=
    \sqrt{V_0^j}X^j-\xi.
\end{equation*}
Since
\[
    \sum_{i,j=1}^nv^{ij}(0)=v(0),
\]
we have
\begin{align}
    &1-\frac1{V_0}
    \sum_{i,j=1}^n
    \Pi_\theta^j\Pi_0^i
    \sqrt{V_0^iV_0^j}\rho^{ij}
    \nonumber =
    -\frac{\sqrt\theta}{v(0)}
    \sum_{i,j=1}^n
    v^{ij}(0)U_\theta^j
    +
    o_{L^1}(\sqrt\theta).
    \label{AB-bracket-expansion}
\end{align}
Together with
\[
    \frac{S_\theta/S_0-1}{\sqrt\theta}
    \Longrightarrow\xi,
\]
this implies
\begin{align*}
    \frac{A_{B,0}^\theta}{\sqrt\theta}
    &\longrightarrow
    -\frac1{v(0)}
    \sum_{i,j=1}^n
    v^{ij}(0)
    \mathsf E[
      U^j1_{\{\xi<0\}}
    ]                                      =
    \frac{\phi(0)}{v(0)^{3/2}}
    \sum_{i,j=1}^n
    v^{ij}(0)\mathsf E[U^j\xi].
\end{align*}
Furthermore,
\[
    \mathsf E[U^j\xi]
    =
    \sum_{k=1}^n
    \frac{v^{jk}(0)}{\Pi_0^j}
    -
    v(0).
\]
Using the symmetry of $v^{ij}(0)$ in the definition of
$\psi_B$, we obtain
\[
    \psi_B
    =
    \frac2{v(0)^{3/2}}
    \sum_{i,j=1}^n
    v^{ij}(0)\mathsf E[U^j\xi].
\]
Consequently,
\begin{equation}\label{AB-limit-psiB}
    A_{B,0}^\theta
    =
    \frac{\phi(0)}2
    \psi_B\sqrt\theta
    +
    o(\sqrt\theta).
\end{equation}
Combining \eqref{AA-limit-psiA} and
\eqref{AB-limit-psiB} proves
\eqref{A-total-expansion}.

The expansion \eqref{Y-total-expansion} is exactly
\eqref{ATM-digital-expansion}, in view of
\eqref{Y-definition} and
$\sqrt{\Sigma_0^S}=\sqrt\theta\,\sigma(0,\theta)$.
The rest is clear from the exact representation \eqref{SSR-A-over-Y}.
\end{proof}


\begin{thebibliography}{99}

\bibitem{AbiJaber}
E. Abi Jaber and S. Li (2026),
Capturing smile dynamics with the quintic volatility model: SPX, SSR and VIX,
\textit{Risk}, May 2026.

\bibitem{A1}
M. Avellaneda, D. Boyer-Olson, J. Busca and P. Friz (2002),
Reconstructing volatility,
\textit{Risk}, October 2002, 91.

\bibitem{A2}
M. Avellaneda, D. Boyer-Olson, J. Busca and P. Friz (2003),
Application of large deviation methods to the pricing of index options in finance,
\textit{Comptes Rendus Math\'ematique} 336(3), 263--266.

\bibitem{BL}
C. Bayer and P. Laurence (2014),
Asymptotics beats Monte Carlo: The case of correlated local volatility baskets,
\textit{Communications on Pure and Applied Mathematics} 67(10), 1618--1657.

\bibitem{RoughBook}
C. Bayer, P. Friz, M. Fukasawa, J. Gatheral, A. Jacquier and M. Rosenbaum (2024),
\textit{Rough Volatility}, SIAM.

\bibitem{Bergomi}
L. Bergomi (2016),
\textit{Stochastic Volatility Modeling}, CRC Press.

\bibitem{EFGR}
O. El Euch, M. Fukasawa, J. Gatheral and M. Rosenbaum (2019),
Short-term at-the-money asymptotics under stochastic volatility models,
\textit{SIAM Journal on Financial Mathematics} 10(2), 491--511.

\bibitem{FZ}
M. Forde and H. Zhang (2016),
Small-time asymptotics for basket options: The bivariate SABR model and the
hyperbolic heat kernel on $\mathbf H^3$,
\textit{SIAM Journal on Financial Mathematics} 7(1), 448--476.

\bibitem{FW}
P. K. Friz and T. Wagenhofer (2023),
Reconstructing volatility: Pricing of index options under rough volatility,
\textit{Mathematical Finance} 33, 19--40.

\bibitem{F21}
M. Fukasawa (2021),
Volatility has to be rough,
\textit{Quantitative Finance} 21, 1--8.

\bibitem{F26}
M. Fukasawa (2026),
On the skew stickiness ratio,
arXiv:2602.05241.

\bibitem{Hagan}
P. Hagan, A. Lesniewski, G. E. Skoufis and D. E. Woodward (2021),
SABR for baskets,
\textit{Wilmott}, March, 50--61.

\bibitem{Pirjol}
D. Pirjol (2023),
Smile consistent basket skew,
\textit{Risk}, September 2023.

\bibitem{Piterbarg}
V. Piterbarg (2006),
Markovian projection method for volatility calibration,
SSRN 906473.

\end{thebibliography}
\end{document}